\documentclass[prodmode,acmtods]{acmsmall}

\usepackage{amsmath}

\usepackage[T1]{fontenc}
\usepackage{graphicx}
\usepackage{booktabs}
\usepackage{subfigure}
\usepackage{amssymb}
\usepackage{url}
\usepackage{cite}
\usepackage{microtype}

\usepackage[ruled, vlined, linesnumbered, nofillcomment]{algorithm2e}

\usepackage{mdwlist}
\usepackage{tikz}
\usetikzlibrary{snakes,arrows,shapes,positioning,fit,calc}
\usepackage{pgfplots}

\newcommand{\todo}[1]{}

\newcommand{\rt}[1]{#1}

\usepackage{etoolbox}
\newcommand{\set}[1]{\left\{#1\right\}}
\newcommand{\pr}[1]{\left(#1\right)}
\newcommand{\fpr}[1]{\mathopen{}\left(#1\right)}
\newcommand{\spr}[1]{\left[#1\right]}

\newcommand{\abs}[1]{{\left|#1\right|}}

\newcommand{\enset}[2]{\left\{#1 ,\ldots , #2\right\}}
\newcommand{\enpr}[2]{\pr{#1 ,\ldots , #2}}

\newcommand{\np}{\textbf{NP}}

\newcommand{\ifam}[1]{\mathcal{#1}}
\newcommand{\marvec}[1]{\mathit{mv}\fpr{#1}}
\newcommand{\mean}[1]{\operatorname{E}\spr{#1}}

\newcommand{\define}{\leftarrow}

\DeclareRobustCommand{\dispfunc}[2]{%
  \ensuremath{%
  \ifthenelse{\equal{#2}{}}%
    {\mathit{#1}}%
    {\mathit{#1}\fpr{#2}}}}

\newcommand{\sm}[1]{\dispfunc{r}{#1}}
\newcommand{\supp}[1]{\dispfunc{sp}{#1}}
\newcommand{\pclos}[1]{\dispfunc{\sigma_c}{#1}}
\newcommand{\pfree}[1]{\dispfunc{\sigma_f}{#1}}
\newcommand{\pndi}[1]{\dispfunc{\sigma_n}{#1}}
\newcommand{\pts}[1]{\dispfunc{\sigma_s}{#1}}
\newcommand{\orf}[1]{\dispfunc{o}{#1}}
\newcommand{\diff}[1]{\dispfunc{d}{#1}}
\newcommand{\cl}[1]{\dispfunc{\mathit{cl}}{#1}}
\newcommand{\tid}{\mathit{tid}}

\SetKwComment{tcpas}{\{}{\}}
\SetCommentSty{textnormal}
\SetArgSty{textnormal}
\SetKw{Break}{break}
\SetKw{Continue}{continue}
\SetKw{False}{false}
\SetKw{True}{true}
\SetKw{Undecided}{undecided}
\SetKw{Null}{null}
\SetKw{AND}{and}
\SetKw{OR}{or}
\SetKwInOut{Input}{input}
\SetKwInOut{Output}{output}
\SetKwInOut{Outputa}{outputa} 

\def\clap#1{\hbox to 0pt{\hss#1\hss}} 
\def\mathclap{\mathpalette\mathclapinternal} 
\def\mathclapinternal#1#2{%
\clap{$\mathsurround=0pt#1{#2}$}%
}

\pgfdeclarelayer{background}
\pgfdeclarelayer{foreground}
\pgfsetlayers{background,main,foreground}

\definecolor{yafaxiscolor}{rgb}{0.3, 0.3, 0.3}

\definecolor{yafcolor1}{rgb}{0.4, 0.165, 0.553}
\definecolor{yafcolor2}{rgb}{0.949, 0.482, 0.216}
\definecolor{yafcolor3}{rgb}{0.47, 0.549, 0.306}
\definecolor{yafcolor4}{rgb}{0.925, 0.165, 0.224}
\definecolor{yafcolor5}{rgb}{0.141, 0.345, 0.643}
\definecolor{yafcolor6}{rgb}{0.965, 0.933, 0.267}
\definecolor{yafcolor7}{rgb}{0.627, 0.118, 0.165}
\definecolor{yafcolor8}{rgb}{0.878, 0.475, 0.686}

\newlength{\yafaxispad}
\newlength{\yaftlpad}
\newlength{\yaflabelpad}
\newlength{\yafaxiswidth}
\newlength{\yafticklen}
\makeatletter
\def\pgfplots@drawtickgridlines@INSTALLCLIP@onorientedsurf#1{}
\makeatother

\newcommand{\yafdrawaxis}[4]{
	\pgfplotstransformcoordinatex{#1}\let\xmincoord=\pgfmathresult 
	\pgfplotstransformcoordinatex{#2}\let\xmaxcoord=\pgfmathresult 
	\pgfplotstransformcoordinatey{#3}\let\ymincoord=\pgfmathresult 
	\pgfplotstransformcoordinatey{#4}\let\ymaxcoord=\pgfmathresult 
	\pgfsetlinewidth{\yafaxiswidth} 
	\pgfsetcolor{yafaxiscolor}
	\pgfpathmoveto{\pgfpointadd{\pgfpointadd{\pgfplotspointrelaxisxy{0}{0}}{\pgfqpointxy{\xmincoord}{0}}}{\pgfqpoint{-0.5\yafaxiswidth}{\yafaxispad}}}
	\pgfpathlineto{\pgfpointadd{\pgfpointadd{\pgfplotspointrelaxisxy{0}{0}}{\pgfqpointxy{\xmaxcoord}{0}}}{\pgfqpoint{0.5\yafaxiswidth}{\yafaxispad}}}
	\pgfpathmoveto{\pgfpointadd{\pgfpointadd{\pgfplotspointrelaxisxy{0}{0}}{\pgfqpointxy{0}{\ymincoord}}}{\pgfqpoint{\yafaxispad}{-0.5\yafaxiswidth}}}
	\pgfpathlineto{\pgfpointadd{\pgfpointadd{\pgfplotspointrelaxisxy{0}{0}}{\pgfqpointxy{0}{\ymaxcoord}}}{\pgfqpoint{\yafaxispad}{0.5\yafaxiswidth}}}
	\pgfusepath{stroke}
}

\pgfplotscreateplotcyclelist{yaf}{%
{yafcolor1,mark options={scale=0.75},mark=o}, 
{yafcolor2,mark options={scale=0.75},mark=square},
{yafcolor3,mark options={scale=0.75},mark=triangle},
{yafcolor4,mark options={scale=0.75},mark=o},
{yafcolor5,mark options={scale=0.75},mark=o},
{yafcolor6,mark options={scale=0.75},mark=o},
{yafcolor7,mark options={scale=0.75},mark=o},
{yafcolor8,mark options={scale=0.75},mark=o},
{black,mark options={scale=0.75},mark=o}} 

\pgfplotsset{axis y line=left, axis x line=bottom,
	tick align=outside,
	tickwidth=\yafticklen,
	clip = false,
    x axis line style= {-, line width = 0pt, color=black!0},
    y axis line style= {-, line width = 0pt, color=black!0},
    x tick style= {line width = \yafaxiswidth, color=yafaxiscolor, yshift = \yafaxispad},
    y tick style= {line width = \yafaxiswidth, color=yafaxiscolor, xshift = \yafaxispad},
    x tick label style = {font=\scriptsize, yshift = \yaftlpad},
    y tick label style = {font=\scriptsize, xshift = \yaftlpad},
    every axis y label/.style = {at = {(ticklabel cs:0.5)}, rotate=90, anchor=center, font=\scriptsize, yshift = -\yaflabelpad},
    every axis x label/.style = {at = {(ticklabel cs:0.5)}, anchor=center, font=\scriptsize, yshift = \yaflabelpad},
    x tick label style = {font=\scriptsize, yshift = 1pt},
    grid = major,
    major grid style  = {dash pattern = on 1pt off 3 pt},
	every axis plot post/.append style= {line width=\yafaxiswidth} ,
	legend cell align = left,
	legend style = {inner sep = 1pt, cells = {font=\scriptsize}},
	legend image code/.code={%
		\draw[mark repeat=2,mark phase=2,#1] 
		plot coordinates { (0cm,0cm) (0.15cm,0cm) (0.3cm,0cm) };%
	} 
}

\title{Finding Robust Itemsets Under Subsampling}

\author{Submitted for Double-Blind review}

\author{NIKOLAJ TATTI
\affil{HIIT, Aalto University, KU Leuven, University of Antwerp}
FABIAN MOERCHEN
\affil{Amazon} 
TOON CALDERS
\affil{Universit\'{e} Libre de Bruxell\'{e}s, Eindhoven University of Technology}}

\begin{abstract}
Mining frequent patterns is plagued by the problem of pattern explosion
making pattern reduction techniques a key challenge in pattern mining.  In this
paper we propose a novel theoretical framework for pattern reduction. We do this
by measuring the robustness of a property of an itemset such as closedness or non-derivability.
The robustness of a property is the probability that this property holds on
random subsets of the original data.
We study four properties: if an itemset is closed, free,
non-derivable or totally shattered, and demonstrate how to compute
the robustness analytically without actually sampling the data.
Our concept of robustness has many advantages: Unlike statistical approaches
for reducing patterns, we do not assume a null hypothesis or any noise model and in contrast to 
noise tolerant or approximate patterns, the robust patterns for a given property are always a 
subset of the patterns with this property.
If the underlying property is monotonic, then the measure is also monotonic,
allowing us to efficiently mine robust itemsets. We further derive a
parameter-free technique for ranking itemsets that can be used for
top-$k$ approaches. Our experiments demonstrate that we can successfully use the
robustness measure to reduce the number of patterns and that ranking yields interesting
itemsets.
\end{abstract}
\category{H.2.8}{Database Management}{Data Mining}
\terms{Algorithms, Experimentation, Theory}
\keywords{pattern reduction, robust itemsets, closed itemsets, free itemsets, non-derivable itemsets, totally shattered itemsets}
\acmformat{}

\begin{document}
\begin{bottomstuff}
The research described in this paper builds upon and extends the work presented
in the IEEE International Conference on Data Mining (ICDM IEEE
2011)~\cite{tatti:11:finding}.

Part of this work done while Nikolaj Tatti was employed by ADReM Research
Group, Department of Mathematics and Computer Science, University of Antwerp
and DTAI group, Department of Computer Science, Katholieke Universiteit Leuven,
Leuven, Belgium. In addition, Fabian Moerchen was employed by Siemens Corporate
Research, USA and Toon Calders was employed by Faculty of Mathematics and Computer
Science, Eindhoven University of Technology, The Netherlands.  Nikolaj Tatti
was partly supported by a Post-Doctoral Fellowship of the Research Foundation
--- Flanders (FWO).

Authors' address: N.\ Tatti, Helsinki Insitute for Information Technology, Department of Information and Computer Science, Aalto University, Finland.
F.\ Moerchen, Amazon, Seattle, Washington, USA.
T.\ Calders, WIT group, Computer \& Decision Engineering department, Universit\'{e} Libre de Bruxell\'{e}s, Belgium
\end{bottomstuff}
\maketitle

\section{Introduction}\label{sec:introduction}
Frequent itemset mining was first introduced in the context of market basket
analysis~\cite{agrawal93mining}. This problem can be defined as follows: a transaction is a subset of
a given set of items $A$, and a transaction database is a set of such transactions.
A subset $X$ of $A$ is a \emph{frequent itemset} in a transaction database if the number 
of transactions containing all items of $X$ exceeds a given threshold. Since its proposal, 
frequent itemset mining has been used to address many data
mining problems such as association rule
generation~\cite{hipp00algorithms}, clustering~\cite{wang99clustering},
classification~\cite{cheng07discriminative}, temporal data
mining~\cite{moerchen10efficient} and outlier detection~\cite{vreeken11odd}.
The mining of itemsets is a core step in these methods that often dominates the
overall complexity of the problem. The number of frequent itemsets, however, can be
extremely large even for moderately sized datasets; in worst case,
the number of frequent itemsets is exponential in $|A|$. 
This explosion severely complicates 
manual analysis or further automated processing steps.

Therefore, researchers have proposed many solutions to reduce the number of patterns
depending on the context in which the patterns are to be used or the process
in which the data was generated. Example of reduced pattern collections include:
 the closed itemsets~\cite{pasquier99discovering} to avoid redundant association rules,
constrained itemsets~\cite{pei01mining} to incorporate prior knowledge,
condensed representations~\cite{calders06survey} to answer frequency queries
with limited memory, margin-closed itemsets~\cite{moerchen10efficient} for
exploratory analysis, and surprising itemsets
\cite{brin:97:beyond,tatti08maximum} or top-k patterns~\cite{geerts04tiling}
for itemset ranking.

Many of reduction techniques have a drawback of being fragile.
For example, a closed itemset can be defined as an itemset that 
can be written as the intersection of transactions; that is, all
of its supersets are contained in strictly less transactions.
Given a non-closed itemset $X$, adding a single transaction to
the dataset containing only $X$ will make $X$ closed.
In this paper we introduce a novel theoretical framework that uses this
drawback to its advantage. Given a property of an itemset (closedness or
non-derivability, for example) we can measure the \emph{robustness} of this
property. A property of $X$ is robust if it holds for many datasets subsampled
from the original data. We demonstrate that we can compute this measure
analytically for several important classes of itemsets:
closed~\cite{pasquier99discovering}, free~\cite{boulicaut03freesets},
non-derivable ~\cite{calders06nonderivable}, and totally shattered
itemsets~\cite{mielikainen05databases}. Computing robust itemsets under
subsampling turns out to be practical for free, non-derivable, and totally
shattered itemsets. Unfortunately, for closed itemsets the test for robustness is
prohibitively expensive.

A possible drawback of our approach is that it depends on a parameter $\alpha$,
the probability of including a transaction in a subsample. In addition to
providing reasonable guidelines to choose $\alpha$ we also introduce a technique
making us independent of $\alpha$. We show that there is a neighborhood near 1 in
which the ranking of itemsets does not depend on $\alpha$. We further
demonstrate how we can compute this ranking without actually discovering the
exact neighborhood or computing the measure for the itemsets. We give exact solutions
for free, non-derivable, and totally shattered itemsets and provide practical heuristics
for closed itemsets.

In the remainder of this paper we describe related work and motivate our
approach in Section~\ref{sec:related}. Itemsets robust under subsampling and
algorithms to find them are described in Section~\ref{sec:theory}.
We discuss ordering itemsets based on large values of $\alpha$ in Sections~\ref{sec:order}--\ref{sec:orderpractice}.
Section~\ref{sec:experiments} demonstrates how the subsampling approach can
reduce the number of reported itemsets significantly. The results are discussed
in comparison with approximate itemsets in Sections~\ref{sec:discussion}.

\section{Related work and motivation}\label{sec:related}
The design goal of condensed representations~\cite{calders06survey} of frequent itemsets is to be able to answer all possible frequency queries. For example, non-derivable itemsets~\cite{calders06nonderivable} exclude any itemset whose support can be derived exactly from the supports of its subsets using logical rules. Other examples of such complete collections are the closed and the free itemsets which are based upon the notion of equivalence of itemsets. Two itemsets are equivalent if they are supported by exactly the same set of transactions. This notion of equivalence divides the frequent itemsets into equivalence classes. The unique maximal element of each equivalence class is a closed itemset~\cite{pasquier99discovering}. No more items can be added to this set without losing some supporting transactions. The not necessarily unique minimal elements of the equivalence class are free itemsets~\cite{boulicaut03freesets} or generators. No items can be taken out without adding transactions to their support set. Complete condensed representations such as those based upon the non-derivable, closed, and free sets allow the derivation of the support of all frequent itemsets. Such representations are useful because they are more compact, yet they still support further mining tasks such as the generation of association rules where the frequencies of all subsets of an itemset are needed to determine the confidence of all possible rules. 

Nevertheless, even the number of closed and free itemsets can still be very large when the minimum support threshold is low. As for other tasks knowing the frequency of all frequent 
itemsets may be less useful because there is a large redundancy in the set of frequent itemsets.  By using approximate methods the number of patterns can be further reduced; for instance by
clustering itemsets representing similar sets of transactions~\cite{xin05mining}, enforcing itemsets to have a minimum margin of difference in support~\cite{moerchen10efficient},
or ranking itemsets by significance~\cite{brin:97:beyond,gallo07mini,webb07discovering,tatti08maximum}.

The above approaches have in common that the complete dataset is considered and
no assumption on potential noise is made. In fault tolerant approaches the
strict definition of support, requiring all items of an itemset to be present
in a transaction is relaxed, see~\cite{gupta08quantitative,calders04mining,uno07efficient,luccese10mining}.
Rather, it is assumed that items can present or absent at random in the transactions.
These approaches can reveal important structures in noisy data that might
otherwise get lost in a huge amount of fragmented patterns. One needs to be
aware though that they report approximate support values and possibly list
itemsets that are not observed as such in the collection at all or with much
smaller support. Also the design goal is not to reduce the number of reported patterns.
Only~\cite{cheng06acclose} considers the3 combination of the two approaches and studies closedness in combination with fault
tolerance.

Furthermore, a third class of techniques considers a statistical null hypothesis and ranks patterns according to 
how much their support deviates from their expected support under the null model~\cite{brin:97:beyond,gallo07mini,webb07discovering,tatti08maximum}.
Unlike these approaches, we do not assume a statistical null hypothesis. We also do not assume any noise
model, such as flipping the values of a matrix independently. Instead our goal
is to study robustness of a given property based on subsampling transactions.

\rt{
The idea of using random databases to assess data mining results has been
proposed in~\cite{gionis:07:assessing,hanhijarvi:09:tell,debie:11:dami}.  The
goal is to first infer some (simple) background information from a dataset, and
then consider all datasets that have the same statistics. A data mining result is then deemed interesting only if it appears in a small number of these datasets.
Interestingly enough, this is the opposite of what we are considering to be important;
that is, we want to find itemsets that satisfy the predicate in many random
subsets of the data. This philosophical difference can be explained by completely
orthogonal randomizations.  The authors in the aforementioned papers sample
random datasets from simple statistics, that is, they ignore on purpose complex
interactions between items, and try to explain mining results with simple
information. Our goal is not to explain results but rather to test whether our
results are robust by testing how data mining results change if we remove
transactions. 

An idea using random datasets to compute the smoothness of results has been
proposed in~\cite{DBLP:conf/pkdd/MisraGT12}. The idea is to measure how stable
the results are by sampling random datasets from a distribution that favors datasets
close to the original one, and computing the average deviation from the original result in the sampled datasets. Finally, stability of rankings has been studied
in the context of networks, see for example~\cite{ghoshai:11:ranking}. 
}

\section{Robust itemsets}\label{sec:theory}
In this section we define the robustness and describe how to compute it efficiently.
\subsection{Notation and definitions}\label{sec:prel}
We begin by reviewing the preliminaries and introducing the notations used in
the paper.

A \emph{binary dataset} $D$ is a set of transactions, tuples $(\tid,
t)$ consisting of a transaction id and a binary vector $t \in \set{0, 1}^K$ of
length $K$.  The $i$th element of a transaction corresponds to an \emph{item}
$a_i$; a $1$ in the $i$th position indicates that the transaction contains the item, a $0$ that it does not. 
We denote the collection of all items by $A= \enset{a_1}{a_K}$.

If $S$ is a set of binary vectors of length $K$, we will write $D \cap S$ to denote
$\set{(\tid, t) \in D \mid t \in S}$.

An itemset $X$ is a subset of $A$. Given a binary vector $t$ of length $K$ and an
itemset $X$, we define $t_X$ to be the binary vector of length $\abs{X}$ obtained by keeping
only the positions corresponding to the items in $X$.

Given an itemset $X = \enpr{x_1}{x_N}$ and a binary vector $v$
of length $N$, we define the \emph{support}
\[
	\supp{X = v; D} = \abs{\set{(\tid, t) \in D \mid t_X = v}}
\]
to be the number of transactions in $D$, where the items in $X$ obtain
the values given in $v$. We often omit $D$ from the notation, when it is clear from the context.
In addition, if $v$ contains only 1s, we simply write $\supp{X}$.
Note that $\supp{X}$ coincides with the traditional definition of a support for
$X$. Discovering frequent itemsets, that is, itemsets whose support exceeds
some given threshold is a well-studied problem.

\begin{example}
Throughout the paper we will use the following dataset $D$ as a running example:
\[
D = \spr{
\begin{smallmatrix}
1: & 0 & 0 & 0 & 0 & 1\\
2: & 0 & 1 & 0 & 1 & 1\\
3: & 1 & 1 & 1 & 1 & 1\\
4: & 0 & 1 & 0 & 1 & 1\\
5: & 1 & 1 & 1 & 1 & 1\\
6: & 1 & 0 & 0 & 0 & 0\\
\end{smallmatrix}}
\]
$D$ contains $5$ items, $a$, $b$, $c$, $d$, and $e$, and $6$ transactions.
For this dataset we have $\supp{ab} = 2$, and $\supp{ab = [1, 0]} = 1$.
\end{example}

We say that a function $f$ mapping an itemset $X$ to a real number $f(X)$ is
\emph{monotonically decreasing} if for each $Y \subseteq X$ we have $f(Y) \geq
f(X)$. \rt{A classic pattern mining task is to discover all itemsets of having $f(X) \geq \rho$
given a threshold $\rho$ and a function $f$ mapping an itemset to a real number. If this function
turns out to be monotonically decreasing, then we can use efficient pattern mining
algorithms to discover \emph{all} patterns satisfying this criterion.}

\rt{
Our next step is to define 4 different properties for itemsets. These are
closed, free, non-derivable, and totally shattered itemsets. The goal of this work is to study how to introduce a measure of robustness for these properties.}

\paragraph{Closed Itemsets}
An itemset $X$ is said to be \emph{closed}, if there is no $Y \supsetneq X$
such that $\supp{X} = \supp{Y}$, i.e., $X$ is maximal w.r.t. set inclusion
among the itemsets having the same support. We define a predicate
\[
	\pclos{X; D} =
	\begin{cases}
		1 & \text{if } X \text{ is closed in } D,\\
		0 & \text{otherwise }\quad .
	\end{cases}
\]
Every closed itemset corresponds to the intersection of a subset of transactions in $D$ and vice versa.

\paragraph{Free Itemsets}
An itemset $X$ said to be \emph{free} if there is no $Y \subsetneq X$
such that $\supp{X} = \supp{Y}$, i.e., free itemsets are minimal
among the itemsets having the same support. We define a predicate
\[
	\pfree{X; D} =
	\begin{cases}
		1 & \text{if } X \text{ is free in } D,\\
		0 & \text{otherwise }\quad .
	\end{cases}
\]

A vital property of free itemsets is that they constitute a downward closed
collection allowing efficient mining with an Apriori-style algorithm (see
Theorem~1 in~\cite{boulicaut00approximationof}). That is, if an itemset $X$ 
is free, all its subsets are free as well.

\begin{example}
The closed itemsets in our running example are $a$, $e$, $bde$, and $abcde$.
On the other hand, the itemsets $\emptyset$, $a$, $b$, $c$, $d$, $e$, $ab$, $ad$, and $ae$ are free.
\end{example}

\paragraph{Non-derivable Itemsets}
An itemset $X$ is said to be \emph{derivable}, if we can derive its support from
the supports of the proper subsets of $X$, otherwise an itemset is called
\emph{non-derivable}. We define a predicate
\[
	\pndi{X; D} =
	\begin{cases}
		1 & \text{if } X \text{ is non-derivable in } D,\\
		0 & \text{otherwise }\quad .
	\end{cases}
\]
Non-derivable itemsets form a downward closed collection (Corollary~3.4~in~\cite{calders06nonderivable}), 
hence we can mine them using an Apriori-style approach.

We say that an itemset $X$ is \emph{totally shattered} if $\supp{X = v} > 0$
for all possible binary vectors $v$.  In other words, every possible
combination of values for $X$ occur in $D$.  Again, we define a predicate
\[
	\pts{X; D} =
	\begin{cases}
		1 & \text{if } X \text{ is totally shattered in } D,\\
		0 & \text{otherwise }\quad .
	\end{cases}
\]
Totally shattered itemsets are related to the
VC-dimension~\cite{mielikainen05databases},
and we can show that a totally shattered itemset is always free and non-derivable (but not the other way around).

\begin{example}
Itemset $ab$ in the running example is totally shattered. Itemset $ac$
is non-derivable but not totally shattered because $\supp{ac = [0, 1]} = 0$.
\end{example}

It is easy to see from the definition that totally shattered itemsets constitute a downward closed collection,
hence they are easy to mine using an Apriori-style approach.

\subsection{Measuring robustness}\label{sec:method}
In this section we propose a measure of robustness for itemsets with a predicate $\sigma$.
\rt{The idea is to sample random subsets from a given dataset and measure how often the predicate $\sigma(X)$
holds in a random dataset.}
Intuitively we consider an itemset robust if the predicate is true for many subsets of the database.

In order to define the measure formally, we first define a probability for a subset of $D$.

\begin{definition}
Given a binary dataset $D$, and a real number $\alpha$, $0 \leq \alpha \leq 1$, we
define a random dataset $D_\alpha$ obtained from $D$ by keeping each transaction
with probability $\alpha$, or otherwise discarding it. More formally, let
$S$ be a subset of $D$. The probability of $D_\alpha = S$ is equal to
\begin{equation}
\label{eq:prob}
	p(D_\alpha = S) = \alpha^\abs{S}(1 - \alpha)^{\abs{D} - \abs{S}} \quad .
\end{equation}
\end{definition}

\rt{We can now define the robustness of an itemset $X$ as the probability of $\sigma(X)$
being true in a random dataset.}

\begin{definition}
Given a binary dataset $D$, a real number $\alpha$, and an itemset predicate
$\sigma$, we define the robustness to be the probability that $\sigma(X;
D_\alpha) = 1$, that is,
\[
	\sm{X; \sigma, D, \alpha} = p(\sigma(X; D_\alpha) = 1) = \sum_{\mathclap{\sigma(X; S) = 1}} p(D_\alpha = S)\quad.
\]
For notational clarity, we will omit $D$ and $\alpha$ when they are clear
from the context.
\end{definition}

\begin{example}
\label{ex:toyts}
Consider itemset $ab$ in our running example. Let $\alpha = 1/3$.
Note that $\supp{ab = [0, 0]} = \supp{ab = [1, 0]} = 1$ and
$\supp{ab = [0, 1]} = \supp{ab = [1, 1]} = 2$.  In order for
$ab$ to still be totally shattered on a subset each of these supports needs to stay greater
than zero. The probability of this event is equal to
\[
	1/3 \times 1/3 \times (1 - 2/3\times 2/3) \times (1 - 2/3\times 2/3) = 25/729,
\]
because for the first two cases we need to sample the single transaction upholding the property and for the other two cases we need to make sure we do not skip both transactions we need to uphold the property.
\end{example}

Our main goal is to mine itemsets for which the robustness measure
exceed some given threshold $\rho$, that is, find all itemsets for which $\sm{X; \sigma, D, \alpha} \geq \rho$.

In order to mine all significant patterns we need to show that the robustness
measure is monotonically decreasing. This is indeed the case if the underlying
predicate is monotonically decreasing. 

\begin{proposition}
Let $\sigma$ be a monotonically decreasing predicate. Then $\sm{X; \sigma, D,
\alpha}$ is also monotonically decreasing.
\end{proposition}

\begin{proof}
Let $Y$ and $X$ be itemsets such that $Y \subset X$. Then
\[
	\sm{X; \sigma, D, \alpha} = \sum_{\mathclap{\sigma(X; S) = 1}} p(D_\alpha = S) \leq \sum_{\mathclap{\sigma(Y; S) = 1}} p(D_\alpha = S) = \sm{Y; \sigma, D, \alpha},
\]
which proves the proposition.
\end{proof}

\rt{
As pointed out in Section~\ref{sec:prel}, predicates for free, non-derivable,
and totally shattered itemsets are monotonically decreasing. However, the predicate for
closedness is not monotonically decreasing.}

\rt{We will finish this section by considering how robustness depends on $\alpha$.}
If we set $\alpha = 1$, then $\sm{X; \sigma, D, \alpha} = \sigma(X; D)$.
Naturally, we expect that when we lower $\alpha$, the robustness would decrease.
This holds for predicates that satisfy a specific property.

\begin{definition}
We say that a predicate $\sigma$ is \emph{monotonic w.r.t. deletion} if for each
itemset $X$, each dataset $D$, and each transaction $t \in D$ it holds that
if $\sigma(X; D) = 0$, then $\sigma(X; D - t) = 0$.
\end{definition}

\begin{proposition}
Let $\sigma$ be a predicate monotonic w.r.t. deletion. Then
	$\sm{X; \sigma, D, \alpha} \leq \sm{X; \sigma, D, \beta}$,
for $\alpha \leq \beta$.
\end{proposition}

\begin{proof}
We will prove the proposition by induction over $\abs{D}$. The proposition
holds trivially for $\abs{D} = 0$. Assume that the theorem holds for $\abs{D} = N$
and let $D$ be a dataset with $\abs{D} = N + 1$.

Fix $t \in D$ and define a new predicate
	$\sigma_t(X; S) = \sigma(X; S \cup \set{t})$,
where $S$ is a dataset. $\sigma_t$ is monotonic w.r.t deletion. Otherwise,
if there is a dataset $S$, a transaction $u \in S$ an itemset $Y$ violating
the monotonicity, then $S \cup \set{t}$, the same transaction $u$, and the itemset $Y$
will violate the monotonicity for $\sigma$.

Moreover, since $\sigma$ is monotonic w.r.t deletion, it holds that
$\sigma(X; S) \leq \sigma_t(X; S)$. This in turns implies that
\begin{equation}
\label{eq:augpred}
	\sm{X; \sigma, S, \alpha} \leq \sm{X; \sigma_t, S, \alpha}\quad.
\end{equation}
Let us write $D' = D - \set{t}$. Then we have,
\[
\begin{split}
	\sm{X; \sigma, D, \alpha}  &= (1 - \alpha)\sm{X; \sigma, D', \alpha} + \alpha\sm{X; \sigma_t, D', \alpha} \\
	&\leq (1 - \beta)\sm{X; \sigma, D', \alpha} + \beta\sm{X; \sigma_t, D', \alpha} \\
	&\leq (1 - \beta)\sm{X; \sigma, D', \beta} + \beta\sm{X; \sigma_t, D', \beta} \\
	&= \sm{X; \sigma, D, \beta},  \\
\end{split}
\]
where the first inequality holds because of Equation~\ref{eq:augpred} and the second
inequality holds because of induction assumption.
This proves the proposition.
\end{proof}

It turns out that all the predicates we considered in Section~\ref{sec:prel} are
monotonic w.r.t. deletion.

\begin{proposition}
\label{prop:monotonedelete}
Predicates $\pclos{}$, $\pfree{}$, $\pndi{}$, and $\pts{}$ are monotonic w.r.t. deletion.
\end{proposition}

\rt{
In order to prove the case for non-derivable itemsets we will need the following technical lemma. 
We will also use this lemma later on.}

\begin{lemma}
\label{lem:ndits}
An itemset $X$ is derivable if and only if there are two
vectors $v$ and $w$ of length $\abs{X}$ with $v$ having odd number of $0$s and
$w$ having even number of $0$s such that $\supp{X = v} = \supp{X = w} = 0$.
\end{lemma}
\begin{proof}
\rt{
Let $O$ be the set of binary vectors of length $\abs{X}$ having odd number of $0$s and 
let $E$ be the set of binary vectors of length $\abs{X}$ having even number of $0$s.

An alternative way of describing non-derivable itemsets is to compute the following
quantities
\[
	u = \supp{X} + \min_{x \in O} \supp{X = x} \quad\text{and}\quad l = \supp{X} - \min_{x \in E} \supp{X = x}\quad.
\]
We can show that
$l \leq \supp{X} \leq u$, both $u$ and $l$ can be computed from proper subsets
of $X$ with the inclusion-exclusion principle (see~\cite{calders06nonderivable}). We also
know that an itemset is derivable if and only if $u = l$
(see~\cite{calders06nonderivable}). This is because we know then that $l = \supp{X} = u$.
Let $v = \arg \min_{x \in O} \supp{X = x}$ and $w = \arg \min_{x \in E} \supp{X = x}$.

This implies that $0 = u - l = \supp{X = v} + \supp{X = w}$, which proves the lemma.}
\end{proof}

\begin{proof}[of Proposition~\ref{prop:monotonedelete}]
An itemset is not totally shattered if there is a binary vector $v$ such that
$\supp{X = v; D} = 0$. This immediately implies that $\supp{X = v; D - \set{t}} = 0$.
Thus $\pts{}$ is monotonic w.r.t. deletion. Similarly, Lemma~\ref{lem:ndits} implies
that $\pndi{}$ is monotonic w.r.t. deletion.

An itemset $X$ is not free, if there is $x \in X$ such that there is no
transaction $u \in D$ for which $u_x = 0$ and $u_y = 1$ for all $y \in X -
\set{x}$. If this holds in $D$, then it holds for $D - \set{t}$. This makes
$\pfree{}$ monotonic w.r.t. deletion. Similarly, an itemset $X$ is not closed,
if there is $x \notin X$ such that there is no transaction $u \in D$ for which
$u_x = 0$ and $u_y = 1$ for all $y \in X$.  If this holds in $D$, then it holds
for $D - \set{t}$. This makes $\pclos{}$ monotonic w.r.t. deletion.
\end{proof}

\begin{example}
The itemset $bd$ is not closed because its superset $bde$ is always observed when $bd$ is observed.
No matter which transaction we delete (one with or without $bde$) this will not change.
Note, however, that $bde$ can become non-closed if transactions 2 and 4 are deleted
because then $abcde$ will have the same support of 2.
\end{example}

\subsection{Computing the measure}\label{sec:analytic}
In this section we demonstrate how to compute the robustness measure for the
predicates. Computing the measure directly
from the definition is impractical since $D$ has $2^\abs{D}$ different subsamples. 
It turns out that computing free, non-derivable, and totally shattered itemsets
has practical formulas while the robustness measure for closed
itemsets has no practical formulation (see Table~\ref{tab:times}).

\begin{table}
\tbl{Computational complexity of robustness and orders. Computing measures is explained in Section~\ref{sec:analytic}.
Computing orders is explained in Section~\ref{sec:order}. $K$ is the number of items, $\abs{\ifam{C}}$ is the number of frequent closed itemsets.}{
\begin{tabular}{lrrr}
\toprule
predicate & measure & order & order estimate \\
\midrule
free & $O(\abs{X})$ & $O(\abs{X})$ & --\\
totally shattered & $O(2^\abs{X})$ & $O(2^{\abs{X}})$ & --\\
closed & $O(2^{K - \abs{X}})$ & $O(2^{K - \abs{X}})$ & $O(\abs{\ifam{C}^2})$\\
non-derivable  & $O(2^\abs{X})$ & $O(\abs{D}{2^\abs{X}})$ & --\\
\bottomrule
\end{tabular}}
\label{tab:times}
\end{table}

\rt{We will first demonstrate how to compute robustness for free and totally shattered itemsets}.
In order to do that we introduce the following function:
Given an itemset $X$ and a set of binary vectors $V \subseteq \set{0, 1}^\abs{X}$
we define
\[
	\orf{X, V, \alpha} = \prod_{v \in V} 1 - (1 - \alpha)^{\supp{X = v}}\quad.
\]
Intuitively, $\orf{X, V, \alpha}$ denotes the probability of the following event: for every vector $v\in V$,
$\supp{X=v;D_\alpha}>0$. Note that since every transaction can support at most one $X=v$, the events
$\supp{X=v;D_\alpha}>0$ are independent from each other. \rt{Note that we can compute $\orf{X, V, \alpha}$ in $O(\abs{V})$
time. Our next step is to show that robustness for free itemsets can be expressed with with $\orf{X, V, \alpha}$
for a certain set of vectors $V$.}

\begin{proposition}
\label{prop:freeanalytic}
Given an itemset $X$, let $V$ be the set of $\abs{X}$ vectors having
$\abs{X} - 1$ ones and one 0.  The robustness of a free itemset is
$\sm{X; \pfree{}, \alpha} = o(X, V, \alpha)$.
\end{proposition}

\begin{proof}
Given an item $x \in X$, define an event
$T_x = \supp{X - \set{x}; D_\alpha} > \supp{X; D_\alpha}$.
$X$ is still free in $D_\alpha$ if $T_x$ is true for all $x \in X$.
$T_x$ is true if and only if $D_\alpha$ contains a transaction $t$ with $t_x = 0$ and $t_y = 1$ for $y \in X - \set{x}$.
There are $\supp{X = v; D}$ such transactions, where $v \in V$ is the vector for which $v_x = 0$.
$p(T_x)$ is the probability of not removing all these transactions, thus
\[
	p(T_x) = 1 - (1 - \alpha)^{\supp{X = v; D}}\quad.
\]
Since each of these transaction is missing only one $x \in X$, there are no common
transactions between different events $T_x$, making them independent. Thus,
we can conclude $ \sm{X; \pfree{}, \alpha} = \prod_{x \in X}p(T_x) = o(X, V, \alpha)$.
\end{proof}

\rt{A similar result also holds for totally shattered itemsets.}

\begin{proposition}
\label{prop:tsanalytic}
Given an itemset $X$, let $V$ be the set of all binary vectors of length $\abs{X}$.
The robustness of a totally shattered itemset is $\sm{X; \pts{}, \alpha} = \orf{X, V, \alpha}$.
\end{proposition}

\begin{proof}
Given a binary vector $v \in V$, define an event
$T_v = \supp{X = v; D_\alpha} > 0$.
$X$ is still totally shattered in $D_\alpha$ if $T_v$ is true for all $v \in V$.
$p(T_v)$ is the probability of not removing all these transactions, thus
	$p(T_v) = 1 - (1 - \alpha)^{\supp{X = v; D}}$.
Again, since no transaction can contribute to different $T_v$ being true, the random variables are independent and we obtain $\sm{X; \pts{}, \alpha} = \prod_{v \in V}p(T_v) = \orf{X, V, \alpha}$.
\end{proof}

Note that the formula in Proposition~\ref{prop:tsanalytic} corresponds directly to
Example~\ref{ex:toyts}.

Let us now consider non-derivable itemsets. The analytic formula is somewhat
more complicated than for free or totally shattered itemsets, although, the
principle remains exactly the same.

\begin{proposition}
\label{prop:ndianalytic}
Given an itemset $X$, let $V$ be the set of binary vectors of length
$\abs{X}$ having odd number of ones. Similarly let $W$ be the set of binary
vectors of length $\abs{X}$ having even number of ones.  The robustness
of a non-derivable itemset is
\[
	\sm{X; \pndi{}, \alpha} = 1 - (1 - \orf{X, \alpha, V})(1 - \orf{X, \alpha, W})\quad.
\]
\end{proposition}

\begin{proof}
Let us define the event $T_V$ to be that there is no $v \in V$ such that $\supp{X = v}
= 0$.  Similarly, let $T_W$ be the event that there is no $w \in W$
such that $\supp{X = w} = 0$.  According to Lemma~\ref{lem:ndits}, an
itemset $X$ is derivable if $T_V$ and $T_W$ are both false.

Using the same argument as with Proposition~\ref{prop:tsanalytic}, we see that
$p(T_V) = o(X, \alpha, V)$.  Similarly, $p(T_W) = o(X, \alpha, W)$.  Since $V \cap
W = \emptyset$, events $T_V$ and $T_W$ are independent. Hence,
$\sm{X; \pndi{}, \alpha}$ is equal to
\[
	1 - p(\lnot T_V \land \lnot T_W) = 1 - (1 - p(T_V)(1 - p(T_W))\quad.
\]
This completes the proof.
\end{proof}

We will now consider closed itemsets. Unlike for the free/totally shattered
itemsets, there is an exponential number of terms in the expression for the robustness. 
The key problem is that
while we can write the robustness in a similar fashion as we did in the proofs
of the previous propositions, the events $\supp{X\cup\{y\}}<\supp{X}$ for all $y\in A\setminus X$,
will no longer be independent, and hence we cannot multiply the probabilities of the individual events.
Indeed, in our running example, $bde$ is a closed itemset. The events $\supp{abde;D_\alpha}<\supp{bde;D_\alpha}$
and $\supp{bcde;D_\alpha}<\supp{bde;D_\alpha}$ are clearly dependent since both events occur in exactly the same subsamples, namely those 
that contain at least one of the transactions 3 and 5.

\begin{proposition}
\label{prop:closedanalytic}
The robustness of a closed itemset is
\[
	\sm{X; \pclos{}, \alpha} = \sum_{Y \supseteq X} (-1)^{\abs{Y} - \abs{X}}(1 - \alpha)^{\supp{X} - \supp{Y}}\quad.
\]
\end{proposition}

\begin{proof}
Given an item $y \notin X$, define an event
$E_y = \supp{X \cup \set{y}; D_\alpha} = \supp{X; D_\alpha}$.
Itemset $X$ is still closed in $D_\alpha$ if all $E_y$ are false, thus
$\sm{X; \pclos{}, \alpha}$ is equal to
\[
 	1 - p\big(\bigvee_{y \notin X} E_y\big)  = \sum_{Z\subseteq (A\setminus X)} (-1)^{\abs{Z}}p\big(\bigwedge_{y \in Z} E_y\big),
\]
where the equality follows from the inclusion-exclusion principle.
Through this transformation we now need to determine the probability of all $E_y$, $y\in Z$ simultaneously being true.
For this all $\supp{X} - \supp{Z \cup X}$ transactions containing $X$ but not $Z$ must have been excluded from $D_\alpha$, hence
\[
	p\big(\bigwedge_{y \in Z} E_y\big) = (1 - \alpha)^{\supp{X} - \supp{Z \cup X}} \quad.
\]
Substituting this above and writing $Y = X \cup Z$ leads to the proposition.
\end{proof}

\begin{example}
In our running example, we have $\supp{bde} = 4$.
This itemset has $3$ superitemsets having the supports $\supp{abde} = \supp{bcde} = \supp{abcde} = 2$.
Hence, the measure $\sm{bde; \pclos{}, \alpha}$ is equal to
\[
	1 - (1 - \alpha)^{4 - 2} - (1 - \alpha)^{4 - 2} + (1 - \alpha)^{4 - 2} = 1 - (1 - \alpha)^2,
\]
where itemsets $bde$, $abde$, $bcde$, and $abcde$ correspond to the terms in the given order.
\end{example}

Unlike with the other predicates, analytic robustness for closed itemsets cannot be be computed
in practice since there are $2^{K - \abs{X}}$ terms in the analytic solution. It turns out that
we cannot do much better as computing robustness is \np-hard.

\begin{proposition}
The following \emph{Robustness of a Closed Itemset (RCI)} problem is \np-hard:
\begin{quote}
For a given database $D$ over the set of items $A$, parameters $\alpha,\rho\in [0,1]$, and itemset $X\subseteq A$, decide if
$\sm{X;\sigma_c,D,\alpha}\geq\rho$.
\end{quote}
\end{proposition}
\begin{proof}
We will reduce the well-known \np-complete vertex cover problem to the RCI problem.
Let $G(V,E)$ be a graph. For every vertex $v\in V$, we will create a unique transaction with identifier $\tid_v$. The set of items over which the transactions will be defined
is the set of edges $E=\{e_1,\ldots,e_K\}$.
Let $t_v=[t_{v1},\ldots,t_{vK}]$ denote the binary vector of length $\abs{E}$ defined as: for all $i=1,\ldots, K$,
\[
	t_{vi} = 1 \quad\text{if and only if}\quad e_i\text{ is not incident with }v\quad.
\]
The transaction database $D$ is now defined as
\[
	D=\set{(tid_v,t_v)\mid v\in V}\quad.
\]

The itemset $X$ in the RCI-problem will be the empty set, $X = \emptyset$. Before we specify $\alpha$ and $\rho$, we show the following property:
\begin{lemma}
Let $S\subseteq D$; $\emptyset$ is closed in $S$ if and only if $V_S =\set{v \in V \mid (\tid_v,t_v) \in S}$ is a vertex cover of $G$.
\end{lemma}

\begin{proof}
If $\emptyset$ is closed in $S$, then for every $e$ there is $t \in D$ such that
$t_e = 0$, otherwise $\supp{e} = \supp{\emptyset}$ .
Hence, for all $e\in E$ there must exist at least one
$v\in V_S$ $t_{ve} = 0$, that is, $e$ must be incident
with $v$. Since $e$ was chosen arbitrary, this implies that every edge in $E$
is covered by at least one node in $V_S$ and hence $V_S$ is a vertex cover of
$G$.
\end{proof}

This relation between the closedness of $\emptyset$ in a subsample $S$ and $V_S$ being a vertex-cover allows us to establish the following relation between the robustness of $\emptyset$ in $D$ and the existence of a vertex-cover of size $k$, that holds for any $\alpha\in[0,1]$.
\begin{lemma}
If $G$ has a vertex cover of size $k$,
\[
	\sm{\emptyset;\sigma_c,D,\alpha}\geq \alpha^k(1-\alpha)^{|D|-k}
\]
otherwise,
\[
	\sm{\emptyset;\sigma_c,D,\alpha}\leq \sum_{j={k+1}}^{|D|}\alpha^j(1-\alpha)^{|D|-j}{|D| \choose j}\quad.
\]
\end{lemma}
\begin{proof}
Indeed, let $\mathit{VC}$ be a vertex cover of $G$, then $\emptyset$ is closed in $S=\{(\tid_v,t_v)~|~v\in \mathit{VC}\}$.
The probability that a randomly selected sample equals $S$ is equal to
\[
	L = \alpha^k(1-\alpha)^{|D|-k},
\]
which is a lower bound on the robustness of $\emptyset$.
Otherwise, if there does not exist a vertex cover of size $k$, this implies that $\emptyset$ is not closed in any subsample $S$ of size $k$ or less. Therefore, the probability mass of all subsamples with at least $k+1$ transactions
\[
	U = \sum_{j={k+1}}^{|D|}\alpha^j(1-\alpha)^{|D|-j}{|D| \choose j}
\]
is an upper bound on the robustness of $\emptyset$.
\end{proof}

The proof now concludes by carefully choosing $\alpha$ such that $U\leq L$, and selecting $\rho$ such that $U\leq \rho\leq L$; in that way, the robustness of the closedness of $\emptyset$
exceeds $L$ and hence $\rho$ if $G$ has a vertex cover of size $k$ or less, and otherwise the robustness is below $U$, and hence also below $\rho$. The last step in the proof is hence
to show that we can always pick $\alpha$ such that $U\leq L$. It can easily be seen that $\alpha=2^{-(|D|+1)}$ satisfies this condition:
Since $\alpha \leq 1/2$, we can now bound $U$ by
\[
	\sum_{j={k+1}}^{|D|}\alpha^j(1-\alpha)^{|D|-j}{|D| \choose j} \leq \sum_{j={k+1}}^{|D|}\alpha^{k + 1}(1-\alpha)^{|D|- k - 1}{|D| \choose j} = 2^{|D|} \alpha^{k+1}(1-\alpha)^{|D|-k-1}\quad.
\]
The right hand-side is smaller than $L$
if and only if $1-\alpha \geq 2^{|D|} \alpha$. Note that for our choice of $\alpha$, we have $1 - \alpha = 1 - \alpha=2^{-(|D|+1)} \geq 1/2 = 2^{\abs{D}}\alpha$.

The binary representation of the numbers $\alpha$ and $\rho$ are polynomial in the size of the original vertex cover problem and the reduction can be carried out in polynomial time.\qed
\end{proof}

\section{Ordering patterns}\label{sec:order}
The robustness measure depends on the parameter
$\alpha$. In this section we propose a parameter-free approach. The idea is to
study how the measure is behaving when $\alpha$ is \emph{close} to $1$. We can
show that there is a (small) neighborhood close to 1, where the \emph{ranking}
of itemsets does not depend on $\alpha$, \rt{that is, there exists $\beta < 1$
such that if $\alpha, \alpha'  \in [\beta, 1]$
$\sm{X; \sigma, D, \alpha} \leq \sm{Y; \sigma, D, \alpha}$ if and only if
$\sm{X; \sigma, D, \alpha'} \leq \sm{Y; \sigma, D, \alpha'}$.}

We will show how compute the ranking in this region, that can be used to select
top-$k$ itemsets by robustness without actually computing the measure or
determining $\beta$.

\rt{In this section we will first give first formal definition, and discuss the
theoretical properties of the ranking. In the next section we demonstrate how we can
compute the order in practice, that is, how to avoid determining $\beta$ and
computing the actual robustness.}

\subsection{Measuring robustness when $\alpha$ approaches $1$}

When $\alpha = 1$ then $D_\alpha = D$ with probability $1$ and the measure is
equivalent to the underlying predicate, providing only a crude ranking:
itemsets that satisfy the predicate vs. itemsets that do not.  If we make
$\alpha$ slightly smaller the measure will decrease a little bit for each
itemset.  The amount of this change will vary from one itemset to another based
on how likely removing only very few transactions will break the predicate for
this itemset. We can use the magnitude of this change to obtain a more
fine-grained ranking by robustness. The key result for this is
that there is a small neighborhood below 1 in which the ranking of itemsets
based on the measure does not depend on $\alpha$.

\begin{proposition}
\label{prop:rankok}
Given a predicate $\sigma$ and a dataset $D$, there exists a number $\beta < 1$
such that
\[
	\sm{X; \sigma, D, \alpha} \leq \sm{Y; \sigma, D, \alpha} \quad\text{ if and only if }\quad \sm{X; \sigma, D, \alpha'} \leq \sm{Y; \sigma, D, \alpha'},
\]
for any itemset $X$ and $Y$ and $\beta \leq \alpha \leq 1$, $\beta \leq \alpha' \leq 1$.
\end{proposition}

\begin{proof}
Fix $X$ and $Y$ and consider
\[
	f(\alpha) = \sm{X; \sigma, D, \alpha} - \sm{Y; \sigma, D, \alpha}\quad.
\]
Since the measure is a finite sum of probabilities that are, according to
Eq.~\ref{eq:prob}, polynomials of $\alpha$, the function $f$ is a polynomial.
This implies that $f$ can have only a finite number of $0$s, of $f = 0$.  Consequently there
is a neighborhood $N = [\beta, 1]$ such that either $f(\alpha) \geq 0$ for any
$\alpha \in N$, or $f(\alpha) \leq 0$ for $\alpha \in N$. Since there is only
a finite number of itemsets, we can take the maximum of all $\beta$s to prove the
theorem.
\end{proof}

Proposition~\ref{prop:rankok} allows us to define an order for itemsets
based on the measure for $\alpha \approx 1$.

\begin{definition}
\label{def:rank}
Given a predicate $\sigma$, and a dataset $D$, we say that $X \preceq_\sigma Y$,
where $X$ and $Y$ are itemsets, if there exists $\beta < 1$ such that
	$\sm{X; \sigma, D, \alpha} \leq \sm{Y; \sigma, D, \alpha}$
for any $\alpha$ such that $\beta \leq \alpha \leq 1$.
Moreover, if
	$\sm{X; \sigma, D, \alpha} < \sm{Y; \sigma, D, \alpha}$ for some
$\alpha \geq \beta$, then we write $X \prec_\sigma Y$.
\end{definition}

Note that Proposition~\ref{prop:rankok} implies that $\preceq_\sigma$ is a total linear order.  That is, we can use this
relation to order itemsets.

\subsection{Properties of the order}

In this section we will study the properties of the order. \rt{Namely, we will show
two properties:
\begin{itemize}
\item We will show in Proposition~\ref{prop:breakdown} that robustness for $\alpha
\approx 1$, essentially measures how many transactions we need to remove in
order to make the predicate fail. The more transactions are needed, the more
robust is the itemset.

\item We will show in Proposition~\ref{prop:asympt} that when we increase the number of
transactions, then a ranking based on robustness \emph{for any fixed} $\alpha$ will become
equivalent with the ranking based on $\prec_\sigma$.
\end{itemize}
}

First, we will need the following key lemma that can be proven by elementary real analysis.

\begin{lemma}
\label{lem:coeff}
Let $f(x) = \sum_{i = 0}^N a_ix^i$ be a non-zero polynomial. Let $k$ be the first index such that $a_k \neq 0$
If $a_k > 0$, then
there is a $\beta > 0$ such that $0 \leq x \leq \beta$ implies $f(x) \geq 0$.
Similarly, if $a_k < 0$, then
there is a $\beta > 0$ such that $0 \leq x \leq \beta$ implies $f(x) \leq 0$.
\end{lemma}

The lemma essentially says that if we express the robustness as a polynomial of
$1 - \alpha$, then we can determine the order by studying the coefficients of
the polynomial.

Our first application of this lemma is a characterization of the order. Assume
two itemsets $X$ and $Y$. Assume that we need to remove $n$ transactions in
order to make the predicate $\sigma(Y)$ fail and that we can fail
$\sigma(X)$ by removing less than $n$ transactions. Then it holds that $X
\prec_\sigma Y$. The following proposition generalizes this idea.

\begin{proposition}
\label{prop:breakdown}
Let $\sigma$ be a predicate, $X$ and $Y$ two itemsets, and $D$ a dataset. Define
a vector $c(X)$ of length $\abs{D}$ such that $c_k(X)$ is the number of subsamples of
$D$ with $\abs{D} - k$ points failing the predicate $\sigma(X)$. Similarly, define $c(Y)$.
Then, $c(X) = c(Y)$ implies that $\sm{X; \sigma, D, \alpha} = \sm{Y; \sigma, D, \alpha}$ for any $\alpha$.
If $c(X)$ is larger than $c(Y)$ in lexicographical order, then $X \prec_\sigma Y$.
\end{proposition}

\begin{proof}
Let us first write the robustness of $X$ using the vector $c_k(X)$. We have,
\[
\begin{split}
	1 - \sm{X; \sigma, D, \alpha} &= p(\sigma(X; D_\alpha) = 0) = \sum_{k = 0}^\abs{D} p(\sigma(X; D_\alpha) = 0, \abs{D_\alpha} = \abs{D} - k) \\
	 & = \sum_{k = 0}^\abs{D} (1 - \alpha)^k\alpha^{\abs{D} - k}\ \sum_{\mathclap{\substack{S \subseteq D \\\abs{S} = \abs{D} - k}}}\ 1 - \sigma(X; S) = \sum_{k = 0}^\abs{D} (1 - \alpha)^k\alpha^{\abs{D} - k} c_k(X)\quad.
\end{split}
\]
If $c_k(X) = c_k(Y)$ it follows immediately that the robustness for $X$ and $Y$ are identical.

Assume now that $c(X)$ is larger than $c(Y)$ in lexicographical order. That it, there is $l$
such that $c_l(X) > c_l(Y)$  and $c_k(X) = c_k(Y)$ for $k < l$. We have
\[
\begin{split}
	\sm{Y; \sigma, D, \alpha} - \sm{X; \sigma, D, \alpha} & = \sum_{k = l}^\abs{D} (1 - \alpha)^k\alpha^{\abs{D} - k} (c_k(X) - c_k(Y)) \\
	                                                      & = (c_l(X) - c_l(Y))(1 - \alpha)^l + f(1 - \alpha),
\end{split}
\]
where $f(x)$ is a polynomial such that the degree of an individual term in $f$ is bigger than $l$.
Lemma~\ref{lem:coeff} now proves the proposition.
\end{proof}

Interestingly enough, if we would define the order based on $\alpha \approx 0$,
then we have a similar result with the difference that instead of deleting
transactions we would be adding them. We would rank $Y$ higher than $X$ if we
can satisfy $\sigma(Y)$ with less transactions than the number of transactions
needed to satisfy $\sigma(X)$.

Ranking itemsets based on how many transactions can be deleted is similar to the breakdown point that measures robustness of statistical estimators. The breakdown point for estimators such as the mean is the number of observations that can be made arbitrarily large before the estimator becomes arbitrarily large as well. The breakdown value of the mean is 1, it becomes infinity as soon as one observation is set to infinity. In contrast the median can handle just under half of the observations to be set to infinity before it breaks down.

We will next show that, in essence, for large datasets the robustness for
\emph{any} $\alpha > 0$ will produce the same ranking as the order defined
for $\alpha$ close to $0$. For this we will consider predicates only of certain type.
\rt{The reason for this is to avoid some pathological predicates, for example, 
$\sigma(X; D) = 1$ if $\abs{D}$ is even, and $0$ otherwise.}

\begin{definition}
Let $\sigma$ be a predicate. Let $K$ be the number of items and let $X$ be an
itemset.  We say that $\sigma$ is a \emph{monotone CNF predicate} if there is a
collection $\set{B_i}_1^L$ of sets of binary vectors of length $K$, (possibly)
depending on $X$ and $K$ such that
\[
	\sigma(X; D) =
	\begin{cases}
		1 & \text{if } D \cap B_i \neq \emptyset \text{ for each } i = 1, \ldots, L, \\
		0 & \text{otherwise}, \\
	\end{cases}
\]
that is, in order to $\sigma(X; D) = 1$, $D$ must contain a transaction from each $B_i$.
\end{definition}

Every predicate we consider in this paper is in fact a monotone CNF predicate.

\begin{proposition}
Predicates $\sigma_c$, $\sigma_f$, $\sigma_n$, and $\sigma_s$ are monotone CNF predicates.
\end{proposition}

\begin{proof}
Fix an itemset $X = x_1\cdots x_N$, and $K$, the total number of items. Let $\Omega = \set{0, 1}^K$ be the
collection of all binary vectors of length $K$.

\paragraph{Free itemsets} Let $B_i = \set{t \in \Omega \mid t_{x_i} = 0,
t_{x_j} = 1, j \neq i, 1 \leq j \leq N}$ for $i = 1, \ldots, N$.  In order to $X$ to be free in
$D$, we must have  $D \cap B_i \neq \emptyset$. Otherwise, $\supp{X} = \supp{X
\setminus \set{x_i}}$, making $X$ not free.

\paragraph{Closed itemsets} Define $K - N$ sets by $B_i = \set{t \in \Omega \mid t_{x_j} = 1, t_i = 0, 1 \leq j, \leq N}$ for $i \notin X$.  $X$ is closed in
$D$ if and only if  $D \cap B_i \neq \emptyset$. Otherwise, $\supp{X} = \supp{X \cup \set{x_i}}$, making $X$ not closed.

\paragraph{Totally shattered itemsets} Define $2^K$ sets by $B_u = \set{t \in \Omega \mid t_{x_j} = u_j, 1 \leq j, \leq N}$ for each $u \in \set{0, 1}^K$.
The proposition follows directly from the definition.

\paragraph{Non-derivable itemsets} Let $C_u = \set{t \in \Omega \mid t_{x_j} = u_j, 1 \leq j, \leq N}$ for each $u \in \set{0, 1}^K$.
Define $4^{K - 1}$ sets by $B_{u, v} = C_u \cup C_v$, where $u, v \in \set{0, 1}^K$, $u$ has odd number of 1s and $v$ has even number of 1s.
The proposition follows directly Lemma~\ref{lem:ndits}.
\end{proof}

\begin{example}
In our running example, an itemset $\mathit{bde}$ is closed if and only $D$
contains at least one transaction from $B_1 = \set{(0, 1, 1, 1, 1), (0, 1, 0, 1, 1)}$
and from $B_2 = \set{(1, 1, 0, 1, 1), (0, 1, 0, 1, 1)}$.  The dataset
does contain $(0, 1, 0, 1, 1)$ making $\mathit{bde}$ closed.
\end{example}

In order to prove the main result we need the following lemma showing that the robustness
of a monotone CNF predicate can be expressed in a certain way.
\rt{We can then exploit this expression in Proposition~\ref{prop:asympt}.}

\begin{lemma}
\label{lemma:cnfform}
Let $\sigma$ be a monotone CNF predicate and let $X$ be an itemset. Let $K$ be the number of itemsets.
Then there is a set of coefficients $\set{c_i}_1^N$ and a collection $\set{S_i}_1^N$ of sets of binary vectors of length $K$
such that
\[
	\sm{X; \sigma, D, \alpha} = \sum_{i = 1}^N c_i (1 - \alpha)^{\abs{D \cap S_i}}\quad.
\]
\end{lemma}

\begin{proof}
Let $S$ be a set of binary vectors of length $K$. The probability of a random subsample $D_\alpha$
not having a transaction from $S$ is equal to
\[
	p(D_\alpha \cap S = \emptyset) = (1 - \alpha)^{\abs{D \cap S}}\quad.
\]
We can rewrite the robustness using the inclusion-exclusion principle,
\[
\begin{split}
	\sm{X; \sigma, D, \alpha} & = 1 - p(\sigma(X; D_\alpha) = 0) = 1 - p(D_\alpha \cap B_1 = \emptyset \lor \cdots \lor D_\alpha \cap B_L = \emptyset) \\
	                          & = 1 - \sum_{i = 1}^L p(D_\alpha \cap B_i = \emptyset) + \sum_{1 \leq i < j \leq L} p(D_\alpha \cap ((B_i \cup B_j) = \emptyset) + \cdots \\
	                          & = 1 - \sum_{i = 1}^L (1 - \alpha)^{\abs{D \cap B_i}} + \sum_{1 \leq i < j \leq L} (1 - \alpha)^{\abs{D \cap (B_i \cup B_j)}} + \cdots
\end{split}
\]
The right-hand side of the equation has the correct form, proving the lemma.
\end{proof}

We are now ready to state the main result of this subsection. Assume that we have a dataset $D$
and we create a new larger dataset $R$ by sampling transactions with
replacement from $D$. The dataset $R$ has the same characteristics as $D$, it
is only larger. Then if we have two itemsets $X$ and $Y$ such that $X
\prec_\sigma Y$, then on average we will have $\sm{X; \sigma, R, \alpha} <
\sm{Y; \sigma, R, \alpha}$ for \emph{any} $\alpha > 0$ assuming that $\abs{R}$
is large enough.

\begin{proposition}
\label{prop:asympt}
Let $\sigma$ be a monotone CNF predicate and let $D$ be a dataset. Let $X$ and $Y$
be itemsets such that $X \prec_\sigma Y$ in $D$. Let $q$ be the empirical distribution
of $D$ and let $R_m$ be a dataset of $m$ random transactions drawn from $q$.
Let $0 < \alpha < 1$.  Then there is $M$ such that
\[
	\mean{\sm{X; \sigma, R_m, \alpha}} < \mean{\sm{Y; \sigma, R_m, \alpha}}
\]
for $m > M$.
\end{proposition}

\begin{proof}
Let us write $\beta = 1 - \alpha$. Lemma~\ref{lemma:cnfform} says that we can write the difference in robustness
as
\[
	\sm{Y; \sigma, R_m, \alpha} - \sm{X; \sigma, R_m, \alpha} = \sum_{i = 1}^N c_i \beta^{\abs{R_m \cap S_i}}
\]
for certain coefficients $\set{c_i}_1^N$ and sets of binary vectors $\set{S_i}_1^N$.
Let $d_k = \sum_{\abs{D \cap S_i} = k} c_i$. Since $X \prec_\sigma Y$,
Lemma~\ref{lem:coeff} implies that there is $l$ such that $d_l > 0$ and $d_k =
0$ for $k < l$.

Let $S$ be a set of binary transactions, and let $k = \abs{S \cap D}$, that is, the probability of generating a random transaction belonging to $S$ is $q(t \in S) = k / \abs{D}$.
We have
\[
	\mean{\beta^{\abs{S \cap R_m}}} = \sum_{j = 0}^m \beta^j {m \choose j} q(t \in S)^j(1 - q(t \in S))^{m - j} = \pr{\beta \frac{k}{\abs{D}} + 1 - \frac{k}{\abs{D}}}^m\quad.
\]
We will write $t_k$ as shorthand for the right-side hand of the equation. Note
that since $\beta < 1$, we have $t_{k + 1} < t_k$. We can write the expected difference between robustness as
\[
	\mean{\sum_{i = 1}^N c_i \beta^{\abs{R_m \cap S_i}}} = \sum_{k = 0}^{\abs{D}} d_k t_k^m = d_lt_l^m + \sum_{k = l + 1}^{\abs{D}} d_k t_k^m = t_l^m \big(d_l + \sum_{k = l + 1}^{\abs{D}} d_k (t_k / t_l)^m\big)\quad.
\]
Since $t_k / t_l < 1$, the terms $(t_k / t_l)^m$ approach $0$ as $m$ goes to
infinity. Hence, there is $M$ such that the sum in the right-hand side of the
equation is larger than $-d_l$ for $m > M$.  This guarantees that the
difference is positive proving the proposition.
\end{proof}

\rt{
This proposition suggests that ranking based on a fixed $\alpha$ and a
parameter-free ranking will eventually agree if the dataset is large enough. In
other words, $\beta$ in Proposition~\ref{prop:rankok} will get smaller (on
average) as the size of the dataset increases. We will see this phenomenon
later on in Propositions~\ref{prop:freerankbeta}~and~\ref{prop:tsrankbeta}.}

\section{Computing order in practice}
\label{sec:orderpractice}
In this section we demonstrate how we can compute the ranking for free,
non-derivable, and totally shattered itemsets and how we can estimate the
ranking for closed itemsets.  For computational complexity see
Table~\ref{tab:times}.

\subsection{Free and totally shattered itemsets}
In this section we will demonstrate that we can compute the order for free and
totally shattered itemsets without finding an appropriate $\alpha$.  We will do this by
analyzing the coefficients of the measure viewed as a polynomial of $1 - \alpha$.

\rt{
Note that for free and totally shattered itemsets these polynomials are
given in Proposition~\ref{prop:freeanalytic} and Proposition~\ref{prop:tsanalytic}.
In order to obtain the coefficients of the polynomial we can simply expand the polynomials.
However,} the polynomials in Proposition~\ref{prop:freeanalytic} and Proposition~\ref{prop:tsanalytic}
are regular enough so that we can compute the order without
expanding the polynomials. In order to do so we need the following definition
for ordering sequences.

\begin{definition}
\label{def:biblio}
Given two non-decreasing sequences $s = s_1, \ldots, s_K$ and $t = t_1, \ldots, t_N$,
we write $s \prec t$ if either there is $s_n < t_n$ and $s_i = t_i$ for all $i < n$
or $t$ is a proper prefix sequence of $s$, that is, $s_i = t_i$ for $i \leq N < K$.
We write $s \preceq t$, if $s = t$ or $s \prec t$.
\end{definition}

The following proposition will allow us to order itemsets without expanding the
polynomials in Propositions~\ref{prop:freeanalytic}--\ref{prop:tsanalytic}.

\begin{proposition}
\label{prop:order}
Assume two polynomials
\[
	f(\alpha) = \prod_{i = 1}^K (1 - (1 - \alpha)^{s_i}) \quad\text{and}\quad g(\alpha) = \prod_{i = 1}^N (1 - (1 - \alpha)^{t_i}),
\]
where $s = s_1, \ldots, s_K$ and $t = t_1, \ldots, t_N$ are non-decreasing sequences
of integers, $s_i, t_i \geq 0$.  If $t \preceq s$, then
there is a $\beta < 1$ such that $\beta \leq \alpha \leq 1$ implies $f(\alpha) \geq g(\alpha)$.
\end{proposition}

\begin{proof}
The case $s = t$ is trivial. Hence we assume that $s \neq t$.
If $s_1 = 0$ or $t_1 = 0$, then $f(\alpha) = 0$ or $g(\alpha) = 0$, and the result follows, hence we will assume that $s_i, t_i > 0$.

Let $\set{a_i}$ and $\set{b_i}$ be coefficients such that
\[
	f(\alpha) = \sum_{i} a_i(1 - \alpha)^i \quad\text{and}\quad g(\alpha) = \sum_{i} b_i(1 - \alpha)^i \quad.
\]
Let ${I}_n$ be the collection of all subsequences of $s$
that sum to $n$,
\[
	{I}_n = \big\{u \mid u \text{ is a subsequence of } s,\  \sum_{i = 1}^{\abs{u}} u_i = n\big\}\quad.
\]
Similarly, let ${J}_n$ be the collection of all subsequences of $t$
that sum to $n$. We can rewrite $f(\alpha)$ as

\[
	f(\alpha) = \sum_{{\scriptstyle u \text{ is a}}\atop{\scriptstyle\text{subseq. of } s}} (-1)^{\abs{u}}(1 - \alpha)^{\sum_i u_i}
\]
which implies that
\[
	a_n = \sum_{u \in {I}_n} (-1)^{\abs{u}} \quad\text{and similarly}\quad
	b_n = \sum_{u \in {J}_n} (-1)^{\abs{u}}\quad.
\]

Assume that $t \prec s$. If $s$ is a prefix sequence of $t$, then
\[
	g(\alpha) = f(\alpha)\prod_{i = {K + 1}}^N (1 - (1 - \alpha)^{t_i}) \leq f(\alpha),
\]
which proves the proposition. Let $n$ be as given in Definition~\ref{def:biblio}.
For every $i < t_n < s_n$, the subsequences in ${I}_i$ and
${J}_i$ contain subsequences from $s$ and $t$ with indices smaller than $n$.
Since $s$ and $t$ are identical up to $n$, then it follows that ${I}_i =
{J}_i$ and consequently $a_i = b_i$. Let $u \in {I}_{t_n}$.
Assume that
$\abs{u} > 1$.  Since, we assume that $s_i > 0$, $u$ is a subsequence of $s_1,\ldots,s_{n - 1}$.
This means that we will find the same subsequence in
${J}_{t_n}$.
Let $A$ be the number of singleton sequences in ${I}_{t_n}$, $A = \abs{\set{u \in I_{t_n} \mid \abs{u} = 1}}$,
and let $B$ be the number of singleton sequences in ${J}_{t_n}$. These singleton sequences
correspond to the entries in $s$ and $t$ having the same value as $t_n$.
Since $s$ and $t$ are identical up to $n$, $s$ does not contain $t_n$ after $s_n$, it holds that
$B > A$. We have now $a_{s_n} - b_{s_n} = B - A > 0$. Lemma~\ref{lem:coeff}
now implies that $f(1 - x) \geq g(1 - x)$, when $x$ is close to $0$.
Write $\alpha = 1 - x$ to complete the proof.
\end{proof}

The polynomials in Propositions~\ref{prop:freeanalytic}--\ref{prop:tsanalytic}
have the form used in Proposition~\ref{prop:order}. Consequently, we can use
the proposition to order itemsets. In order to do that we need the following
definitions.

\begin{definition}
Given a dataset $D$ and an itemset $X$, we define a \emph{free margin vector}
$\marvec{X; D, \pfree{}}$ to be the sequence of $\abs{X}$ integers $\supp{X = v; D}$,
where $v$ is a binary vector having $\abs{X} - 1$ ones, \emph{ordered} in the increasing order.

Similarly, we define a \emph{totally shattered margin vector}
$\marvec{X; D, \pts{}}$ to be a sequence of $2^\abs{X}$ integers $\supp{X = v; D}$
\emph{ordered} in the increasing order.
\end{definition}

\begin{corollary}
Given itemsets $X$ and $Y$ and a dataset $D$,
$X \preceq_{\pfree{}} Y$ if and only if $\marvec{X; D, \pfree{}} \preceq \marvec{Y; D, \pfree{}}$.
\end{corollary}

\begin{corollary}
Given itemsets $X$ and $Y$ and a dataset $D$,
	$X \preceq_{\pts{}} Y$ if and only if $\marvec{X; D, \pfree{}} \preceq \marvec{Y; D, \pts{}}$.
\end{corollary}

\begin{example}
In our running example, $\supp{ab = [1, 0]} = 1$ and $\supp{ab = [0, 1]} = 2$,
hence the free margin vector is equal to $\marvec{ab; \pfree{}} = [1, 2]$.
Similarly, we have $\supp{ae = [1, 0]} = 1$ and  $\supp{ae = [0, 1]} = 3$,
hence the free margin vector is equal to $\marvec{ae; \pfree{}} = [1, 3]$.
Hence, we conclude that $ab \prec_{\pfree{}} ae$.
\end{example}

\rt{Margin vectors are useful to determine the order of robust itemsets.
However, we can also use them to provide a bound for $\beta$ given in Definition~\ref{def:rank}.
More specifically, the further the margin vectors are from each other
the lower $\alpha$ can be 
such that the robustness still agrees with the order. To make this formal, we will need the following
definition.}

\begin{definition}
Assume two non-decreasing sequences $s = s_1, \ldots, s_K$ and $t = t_1, \ldots, t_N$
such that $s \preceq t$. Let $n$ be the first index such that $s_n < t_n$, we
define $\diff{s, t} = t_n - s_n$. If no such such index exist, that is,
$t$ is a prefix sequence of $s$, we define $\diff{s, t} = \infty$.
\end{definition}

The following propositions state that the larger $\diff{s, t}$, the lower $\alpha$ can be.
This reflects the result of Proposition~\ref{prop:asympt}: large datasets will result in
large differences in margin vectors, allowing $\alpha$ to be small.

\begin{proposition}
\label{prop:freerankbeta}
Assume itemsets $X$ and $Y$ and a dataset $D$ such that $X \preceq_{\pfree{}}
Y$. Let $d = \diff{\marvec{X; D, \pfree{}}, \marvec{Y; D, \pfree{}}}$.
Then
\[
	\sm{X, \pfree{}, D, \alpha} \leq \sm{Y, \pfree{}, D, \alpha} \quad\text{for}\quad \alpha \geq 1 - \frac{1}{\sqrt[d]{\abs{Y} + 1}}\quad.
\]
\end{proposition}

\begin{proposition}
\label{prop:tsrankbeta}
Assume itemsets $X$ and $Y$ and a dataset $D$ such that $X \preceq_{\pts{}}
Y$. Let $d = \diff{\marvec{X; D, \pts{}}, \marvec{Y; D, \pts{}}}$.
Then
\[
	\sm{X, \pfree{}, D, \alpha} \leq \sm{Y, \pfree{}, D, \alpha} \quad\text{for}\quad \alpha \geq 1 - \frac{1}{\sqrt[d]{2^\abs{Y} + 1}}\quad.
\]
\end{proposition}

Both propositions follow immediately from the following proposition.

\begin{proposition}
Given two non-decreasing sequences $s = s_1, \ldots, s_K$ and $t = t_1, \ldots, t_N$ such that $s \preceq t$,
let $d = \diff{s, t}$.
Then
\[
	\prod_{i = 1}^K (1 - (1 - \alpha)^{s_i}) \leq \prod_{i = 1}^N (1 - (1 - \alpha)^{t_i}) \quad\text{for}\quad \alpha \geq 1 - \frac{1}{\sqrt[d]{N + 1}}\quad.
\]
\end{proposition}

\begin{proof}
If $t$ is a prefix sequence of $s$, then the inequality holds for any $\alpha$.
Assume that $t$ is not a prefix sequence and let $n$ be the first index such
that $s_n < t_n$.
Write $\beta = 1 - \alpha$.  We can upper bound the left-hand side by
\[
	\prod_{i = 1}^K (1 - \beta^{s_i}) \leq (1 - \beta^{s_n})\prod_{i = 1}^{n - 1} (1 - \beta^{s_i})
\]
and lower bound the right-hand side by
\[
	\prod_{i = 1}^N (1 - \beta^{t_i}) \geq (1 - \beta^{s_n + d})^N \prod_{i = 1}^{n - 1} (1 - \beta^{s_i})\quad.
\]
Hence it is sufficient to show that
\[
	(1 - \beta^{s_n}) \leq (1 - \beta^{s_n + d})^N \quad\text{or}\quad \log(1 - \beta^{s_n}) \leq N \log(1 - \beta^{s_n + d})\quad.
\]
We apply the inequalities $-x \leq \log (1 - x) \leq -x /(1 - x)$ which gives us
\[
	-\beta^{s_n} \leq -N\frac{\beta^{s_n + d}}{(1 - \beta^{s_n + d})} \quad\text{or}\quad 1 \geq N\beta^d + \beta^{s_n + d} \quad.
\]
Since $\beta^d \geq \beta^{s_n + d}$ it is sufficient to have $1 \geq (1 + N)\beta^d$.
This is true for $\beta \leq \frac{1}{\sqrt[d]{N + 1}}$.
\end{proof}

\subsection{Closed itemsets}
\label{sec:orderclosed}
In this section we will introduce a technique for estimating the ranking for
closed itemsets. As the measure for closed itemsets has a different form than
for free or totally shattered itemsets we are forced to seek for alternative
approaches. \rt{We approach the problem by first expressing the coefficients of
the polynomial with supports of closed itemsets. Then we estimate the
polynomial by considering only the most frequent closed itemsets.}

Let us consider Proposition~\ref{prop:closedanalytic}.  Let $a_k$ be the
coefficient for the $k$th term of the polynomial for $\sm{X; \pclos{}, \alpha}$
given in Proposition~\ref{prop:closedanalytic}.  If we can compute these
numbers efficiently, we can use Lemma~\ref{lem:coeff} to find the ranking.

We will do this by first expressing $a_k$ using closed itemsets. In order to do
that let $\cl{X}$ be the closure of an itemset $X$.  Let us define
\[
	e(Y, X) = \sum_{Z \supseteq X, \atop \cl{Z} = Y} (-1)^{\abs{Z} + \abs{X}}
\]
to be the alternating sum over all itemsets containing $X$ and having $Y$ as their closure.
Since all the itemsets having the same closure will have the same support we can
write the coefficients $a_k$ using $e(Y, X)$,
\begin{equation}
\label{eq:closedcoeff}
	a_k = \ \sum_{\mathclap{Y \supseteq X, \atop \supp{X} - \supp{Y}  = k}} \  (-1)^{\abs{Y} + \abs{X}}
	= \ \sum_{\mathclap{Y \supseteq X, Y = \cl{Y} \atop \supp{X} - \supp{Y}  = k}} \ e(Y, X)\quad.
\end{equation}

To compute $e(Y, X)$, first note that $e(X, X) = 1$. If $Y \neq X$, then using
the following identity
\[
	\sum_{Y \supseteq Y' \supseteq X \atop Y' = \cl{Y'}} e(Y', X) = \sum_{Y \supseteq Z \supseteq X} (-1)^{\abs{Z} + \abs{X}} = 0
\]
we arrive to
\begin{equation}
\label{eq:closedupdate}
	e(Y, X) = -\sum_{\mathclap{Y \supsetneq Y' \supseteq X \atop Y' = \cl{Y'}}} e(Y', X)\quad.
\end{equation}

Thus, we can compute $e(Y, X)$ from $e(Y', X)$, where $Y'$ is a closed subset of $Y$.
This is convenient, because when computing $e(Y, X)$, say for $a_k$, we
have already computed all the subsets of $Y$ for previous coefficients.

\begin{example}
Consider itemset $e$ in our running example. There are two closed supersets of
$e$, namely $bde$ and $abcde$, having the supports $4$ and $2$, respectively.
Using the update equations, we see that $e(e, e) = 1$, $e(bde, e) = -1$, and
$e(abcde, e) = 0$.  As $\supp{e} = 5$, we see that the non-zero coefficients
${a_i}$ are $a_0 = 1$ and $a_1 = -1$.
\end{example}

The problem with this approach is that we can still have an exponential number
of closed itemsets. Hence, we chose to estimate the ranking by only using
\emph{frequent} closed itemsets and estimate the remaining itemsets to have a
support of $0$.

This estimation is achieved by removing all closed non-frequent itemsets from
the sums of Eqs.~\ref{eq:closedcoeff}~and~\ref{eq:closedupdate} and adding an
itemset containing all the items and having the support $0$. The code for this
estimation is given in Algorithm~\ref{alg:closedestimate}.

\begin{algorithm}[htb!]
\Input{$X$ an itemset, $\ifam{C}$, frequent closed itemsets}
\Output{$\set{a_k}$, coefficients of the polynomial}
\lIf{$A \notin \ifam{C}$} {
	add $A$ to $\ifam{C}$ with $\supp{A} = 0$
}
$\ifam{C} \define \set{Y \in \ifam{C} \mid X \subseteq Y}$ \;
$\ifam{L} \define $ sets in $\ifam{C}$ ordered by the subset relation\;
$e(X, X) \define 1$\;
\For {$Y \in \ifam{L}$} {
	$e(Y, X) \define - \sum_{Z \in \ifam{C}, Z \subsetneq Y} e(Z, X)$\;
	$k \define \supp{X} - \supp{Y}$\;
	$a_k \define a_k + e(Y, X)$\;
}
\caption{Algorithm for estimating coefficients of the polynomial given in Proposition~\ref{prop:closedanalytic}.}
\label{alg:closedestimate}
\end{algorithm}

Algorithm~\ref{alg:closedestimate} takes $O(\abs{\ifam{C}}^2)$ time.  In
practice, this is much faster because an average itemset does not have that
many supersets.

Now that we have a way of estimating $a_k$ from frequent closed itemsets, we
can, given two itemsets $X$ and $Y$, search the smallest $k$ for which the
coefficients differ in order to apply Lemma~\ref{lem:coeff}. Note that if the
index of the differing coefficient, say $k$, is such that $\supp{X} - k$ is
larger or equal to the support threshold, then $a_k$ is correctly computed by
our estimation, and our approximation yields a correct ranking.

\subsection{Non-derivable itemsets}
In this section we will discuss how to compute the ranking non-derivable
itemsets.  The ranking for non-derivable is particularly difficult because we
cannot use Proposition~\ref{prop:order} to avoid expanding the polynomial given
in Proposition~\ref{prop:ndianalytic}. We can, however, expand the polynomial
since, due to Eq.~\ref{eq:prob}, it only has $\abs{D}$ terms. Once we have
expanded the polynomial, we can use Lemma~\ref{lem:coeff} to compare the
itemsets.

First note that we can rewrite the measure as
\begin{equation}
\label{eq:ndiexpand}
	\sm{X; \pndi{}, \alpha} = \orf{X, \alpha, V} + \orf{X, \alpha, W} - \orf{X, \alpha, U},
\end{equation}
where $U$ consists of all binary vectors of length $\abs{X}$, $V$ is the subset
of $U$ containing vectors having odd number of ones, and $W = U \setminus V$.

Next, we will show how to expand a term $\orf{X, \alpha, S}$ for any set of binary
vectors $S$. Once we are able to do that, we can expand each term in
Eq.~\ref{eq:ndiexpand} individually to compute the final coefficients.
In order to do that, we will use the identity
\[
	(1 - x^a)\sum_{i = 0}^N c_i x^{i} = \sum_{i = 1}^{N + a} (c_i - c_{i - a}) x^i,
\]
where in the right-hand side we define $c_i = 0$ for $i < 0$ or $i > N$. This
gives us a simple iterative procedure, given in Algorithm~\ref{alg:expand}: For
each $v \in S$, we shift the current coefficients by $\supp{X = v}$ and
subtract the result from the current coefficients.

\begin{algorithm}[htb!]
\caption{\textsc{Expand}$(X, D, S)$, expands the polynomial $\orf{X, D, S}$}
\label{alg:expand}
\Input{$X$, an itemset, $D$, a dataset, $S$ a set of vectors}
\Output{$\set{c_i}_1^\abs{D}$ set of coefficients of the polynomial $\orf{X, D, S}$}

$c_i \define 0$ for $i = 0, \ldots, \abs{D}$\;
$c_0 \define 1$\;

\ForEach{$v \in S$} {
	$s \define \supp{X = v}$\;
	$n_i \define 0$ for $i = 0, \ldots, s - 1$\;
	$n_i \define c_{i - s}$ for $i = s, \ldots, \abs{D}$\;
	$c_i \define c_i - n_i$ for $i = 0, \ldots, \abs{D}$\;
}
\Return $\set{c_i}_1^\abs{D}$\;
\end{algorithm}

The highest degree in the polynomial will be $\sum_{v \in S} \supp{X = v}$.
Since, each $v$ is unique in $S$, this number is bounded by $\abs{D}$.  This
means that we have to consider only $\abs{D}$ coefficients and that the
computational complexity of \textsc{Expand} is $O(\abs{S}\abs{D})$.
Consequently, computing the coefficients in Eq.~\ref{eq:ndiexpand} will take
$O(\abs{U}\abs{D}) = O(2^\abs{X}\abs{D})$ time. We can further speed this up by
using sparse vectors, and computing the terms in a lazy fashion during the
comparison.

\begin{example}
Consider itemset $ac$ in our running example.  We have $\supp{ac = (0, 0)} = 3$,
$\supp{ac} = 2$, $\supp{ac = (1, 0)} = 1$, and $\supp{ac = (0, 1)} = 0$.
Let $V = \set{(0, 1), (1, 0)}$, $W= \set{(0, 0), (1, 1)}$ and $U = V \cup W$.
Since $\supp{ac = (0, 1)} = 0$, both $\orf{X, \alpha, V}$ and $\orf{X, \alpha, U}$ are $0$.
We have
\[
	\orf{X, \alpha, W} = (1 - (1 - \alpha)^3)(1 - (1 - \alpha)^2) = 1 - (1 - \alpha)^2 - (1 - \alpha)^3 + (1 - \alpha)^5\quad.
\]
Consequently, \textsc{Expand} will return $(1, 0, -1, -1, 0, 1, 0)$ as coefficients.
\end{example}

\section{Experiments}\label{sec:experiments}
In this section we present our experiments. 
\rt{
\begin{itemize}
\item We study typical behavior of robustness for free, totally shattered, and non-derivable itemsets as a function of $\alpha$.
\item We test how similar the rankings are based on robustness and based on the order $\prec_\sigma$.
\item We test how the ranking of robust closed itemsets changes under the effect of noise.
\end{itemize}
In addition, we provide examples of top-k robust closed and free itemsets.}
\subsection{Datasets}

We used datasets from three repositories. The 8 FIMI~\cite{goethals03fimi} datasets
include large transaction datasets derived from traffic data, census data,
and retail data. Two datasets are synthetically generated to simulate market basket data.
The datasets from the UCI Machine Learning Repository~\cite{asuncion07uci}
represent classification problems from a wide variety of domains.
We used the itemset representations of 29 datasets from the LUCS repository~\cite{coenen03lucs}.
Finally we used 18 text datasets
shipped with the Cluto clustering toolkit~\cite{zhao02evaluation} but converted to itemsets
using a binary representation of words in documents discarding the term frequencies.

\subsection{Reducing the number of patterns}

The goal of the first experiment is to show that this new constraint for itemsets
can significantly reduce the number of itemsets reported in the results by removing
itemsets that are spurious in the sense that they are unlikely to be observed on
many subsamples. Throughout this section we will use
$\alpha$ for the size of the data sample,
$\rho$ for the minimum robustness threshold,
and $\tau$ for the minimum support threshold.

Our first question is how the parameters should be chosen.
It is clear that if we choose $\alpha$ very close to 1,
then even itemsets that would lose their predicate by removing
only a few transactions still have a high likelihood of being found. We would
thus expect most robustness values to be close to 1 when $\alpha$ is close to $1$.
This would make choosing a suitable $\rho$ very difficult
and might lead to problems due to floating point arithmetics.
Similarly, choosing $\alpha$ close to 0 will cause most itemsets to have
a very low likelihood of still being found, thus most robustness values will be close to 0.
Thus choosing a medium $\alpha$ will be most useful to emphasize the quantitative difference between itemsets of various
robustness.

As for the minimum robustness threshold $\rho$, the larger its value is, the stricter the filtering will be. Choosing the threshold is somewhat application dependent but it should not be close to zero, otherwise no reduction will be observed.

To confirm our reasoning we performed a parameter study for the itemset version of the \emph{Zoo} dataset that describes
101 animals with 42 boolean attributes. This data contains $9\,702$ free itemsets, $3\,476$
non-derivable itemsets, and $1\,252$ totally shattered itemsets (at minimum support $\tau = 0.01$).  The number of itemsets
as a function of $\alpha$ and $\rho$ is given in Figure~\ref{fig:zoo_matrix}.
As expected
\begin{itemize}
\item for large $\alpha$ all but the largest $\rho$ do not reduce the number of itemsets reported,
\item as $\alpha$ becomes smaller, the itemsets are spread smoothly across the range of $\rho$ allowing a meaningful quantitative evaluation,
\item for small $\alpha$ almost no itemsets are reported even for very small $\rho$.
\end{itemize}

\begin{figure}[htb]
\centering
\subfigure[free itemsets]{\includegraphics{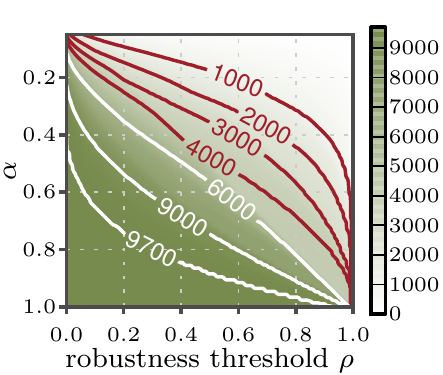}}
\subfigure[non-derivable itemsets]{\includegraphics{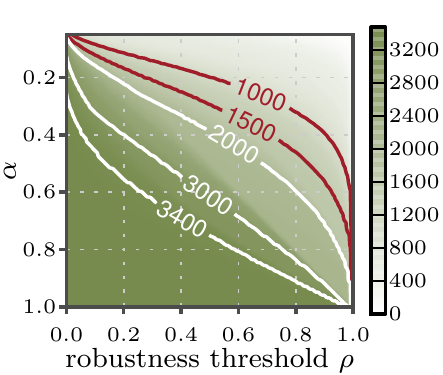}}
\subfigure[totally shattered itemsets]{\includegraphics{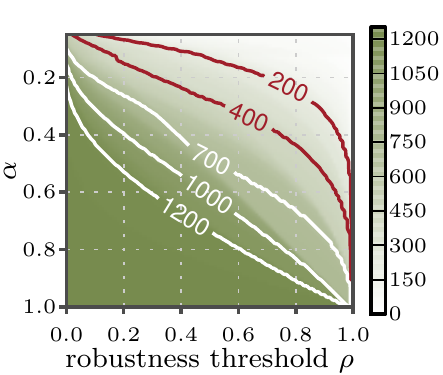}}
\caption{Number of free, non-derivable, and totally shattered itemsets on \emph{Zoo} ($\tau = 0.01$) dataset as a function of $\alpha$ and $\rho$.}
\label{fig:zoo_matrix}
\end{figure}

In order to evaluate if this holds for more datasets, we computed the number of free/non-derivable/totally shattered
using different $\alpha$s and normalized this by the number of robust itemsets
exceeding the minimal robustness threshold of $0.1$. In order to minimize the
variance of behavior of the robustness in a single dataset, we consider an
average over all test datasets, which we give in
Figure~\ref{fig:average}.  We see the same phenomenon as in
Figure~\ref{fig:zoo_matrix}. Large values of $\alpha$ induce a skewed
distribution which becomes more balanced as we decrease the value of $\alpha$.
Consider $\alpha = 0.9$.  Our test datasets typically contain a lot of
itemsets having only one transaction keeping them from becoming non-free. This can be seen
as a dip of the curve for $\alpha = 0.9$ at $\rho = 0.9$ in
Figure~\ref{fig:average:free}.  A second dip at $\rho = 0.81$ represents the
itemsets that can be made non-free by deleting two transactions.
As we make $\alpha$ smaller, these dips become less prominent. 

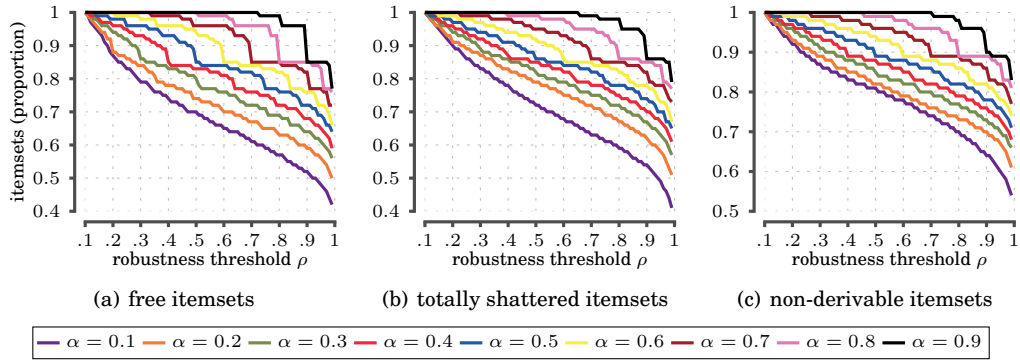
\begin{figure}[htb!]
\begin{center}
\subfigure[free itemsets\label{fig:average:free}]{
\begin{tikzpicture}
\begin{axis}[xlabel={robustness threshold $\rho$}, ylabel= {itemsets (proportion)},
    width = 4.9cm,
    xmax = 1, xmin = 0.1,
	ymax = 1, ymin = 0.4,
	ytick = {0.4, 0.5, ..., 1.1},
	xtick = {0.1, 0.2, ..., 1},
	xticklabels = {$.1$, $.2$, $.3$, $.4$, $.5$, $.6$, $.7$, $.8$, $.9$, $1$},
	no markers,
    cycle list name=yaf,
    ]
\foreach \i in {1, 2, ...,  9} {\addplot table[x index = 0, y expr = {1 - \thisrowno{\i}}] {hists.dat}; }
\pgfplotsextra{\yafdrawaxis{0.1}{1}{0.4}{1}}
\end{axis}
\end{tikzpicture}}
\subfigure[totally shattered itemsets]{
\begin{tikzpicture}
\begin{axis}[xlabel={robustness threshold $\rho$}, 
    width = 4.9cm,
    xmax = 1, xmin = 0.1,
	ymax = 1, ymin = 0.4,
	ytick = {0.4, 0.5, ..., 1.1},
	xtick = {0.1, 0.2, ..., 1},
	xticklabels = {$.1$, $.2$, $.3$, $.4$, $.5$, $.6$, $.7$, $.8$, $.9$, $1$},
	no markers,
    cycle list name=yaf,
    ]
\foreach \i in {10, 11, ...,  18} {\addplot table[x index = 0, y expr = {1 - \thisrowno{\i}}] {hists.dat}; }
\pgfplotsextra{\yafdrawaxis{0.1}{1}{0.4}{1}}
\end{axis}
\end{tikzpicture}}
\subfigure[non-derivable itemsets]{
\begin{tikzpicture}
\begin{axis}[xlabel={robustness threshold $\rho$}, 
    width = 4.9cm,
    xmax = 1, xmin = 0.1,
	ymax = 1, ymin = 0.5,
	ytick = {0.5, 0.6, ..., 1.1},
	xtick = {0.1, 0.2, ..., 1},
	xticklabels = {$.1$, $.2$, $.3$, $.4$, $.5$, $.6$, $.7$, $.8$, $.9$, $1$},
	no markers,
    cycle list name=yaf,
	legend columns = 9,
	legend entries = {$\alpha = 0.1$, $\alpha = 0.2$, $\alpha = 0.3$, $\alpha = 0.4$, $\alpha = 0.5$, $\alpha = 0.6$, $\alpha = 0.7$, $\alpha = 0.8$, $\alpha = 0.9$},
	legend to name = leg:hist
    ]
\foreach \i in {19, 20, ...,  27} {\addplot table[x index = 0, y expr = {1 - \thisrowno{\i}}] {hists.dat}; }
\pgfplotsextra{\yafdrawaxis{0.1}{1}{0.5}{1}}
\end{axis}
\end{tikzpicture}}
\ref{leg:hist}
\end{center}
\caption{Average of the number of free/totally shattered/non-derivable itemsets as a function of $\rho$ normalized by the number of itemsets for $\rho = 0.1$. Average is taken over all test datasets}
\label{fig:average}
\end{figure}

Based on this we chose $\alpha = 0.5$
and plotted the number of free itemsets
as a function of $\rho$.
Figure~\ref{fig:free_zoo_histogram} shows that for the \emph{Zoo} dataset there are many free
itemsets with very different robustness values showing a rich structure that can be exploited
to rank and reduce the number of itemsets. Similar results were observed for many of the UCI datasets.
Figure~\ref{fig:free_la12_histogram} shows a representative example for the text datasets.
While the distribution is much more skewed, a large $\rho$ would
also reduce the number of itemsets by about 50\%.
Finally, Figure~\ref{fig:free_retail_histogram} shows an example for a large transactional dataset
with 88k transactions. Using $\alpha = 0.5$ generated a distribution where all values were close to one
so we needed to set $\alpha = 0.01$ to better show the quantitative differences of the itemsets.
This demonstrates that the more transactions a dataset contains, the more skewed the distribution for a
fixed $\alpha$ will be.

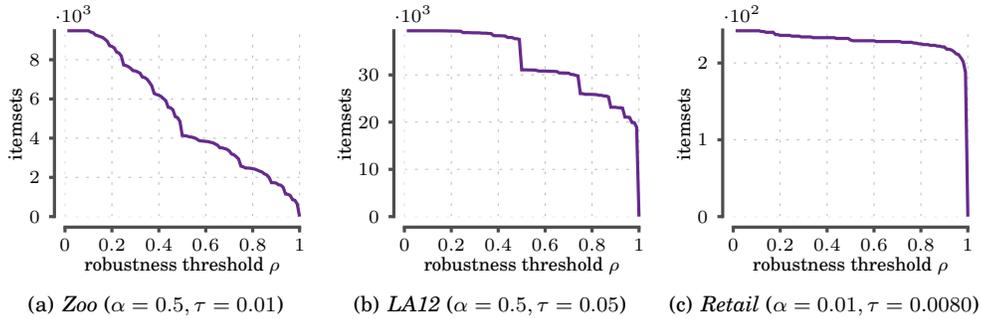
\begin{figure}[htb!]
\centering
\subfigure[\emph{Zoo} ($\alpha=0.5, \tau = 0.01$)\label{fig:free_zoo_histogram}]{
\begin{tikzpicture}
\begin{axis}[xlabel={robustness threshold $\rho$}, ylabel= {itemsets},
    width = 4.7cm,
    xmax = 1, xmin = 0,
	no markers,
    scaled x ticks = false,
    scaled y ticks = base 10:-3,
    cycle list name=yaf,
    ]
\addplot table[x index = 0, y expr = {9479 - \thisrowno{1}}, header = false] {zoo_hist.dat};
\pgfplotsextra{\yafdrawaxis{0}{1}{0}{9479}}
\end{axis}
\end{tikzpicture}}
\subfigure[\emph{LA12} ($\alpha=0.5, \tau = 0.05$)\label{fig:free_la12_histogram}]{
\begin{tikzpicture}
\begin{axis}[xlabel={robustness threshold $\rho$}, ylabel= {itemsets},
    width = 4.7cm,
    xmax = 1, xmin = 0,
	no markers,
    scaled x ticks = false,
    scaled y ticks = base 10:-3,
    cycle list name=yaf,
    ]
\addplot table[x index = 0, y expr = {39350 - \thisrowno{1}}, header = false] {la12_hist.dat};
\pgfplotsextra{\yafdrawaxis{0}{1}{0}{39350}}
\end{axis}
\end{tikzpicture}}
\subfigure[\emph{Retail} ($\alpha=0.01, \tau = 0.0080$)\label{fig:free_retail_histogram}]{
\begin{tikzpicture}
\begin{axis}[xlabel={robustness threshold $\rho$}, ylabel= {itemsets},
    width = 4.7cm,
    xmax = 1, xmin = 0, ymin=0,
	no markers,
    scaled x ticks = false,
    scaled y ticks = base 10:-2,
    cycle list name=yaf,
    ]
\addplot table[x index = 0, y expr = {243 - \thisrowno{1}}, header = false] {retail_hist.dat};
\pgfplotsextra{\yafdrawaxis{0}{1}{0}{243}}
\end{axis}
\end{tikzpicture}}
\caption{Number of free itemsets as a function of $\rho$}
\label{fig:histogram}
\end{figure}

\subsection{Effect of noise for robust closed itemsets}
\rt{
Our next experiment is to see how robust closed itemsets behave when a dataset is
exposed to noise.  Our expectation is that most robust itemsets will stay
closed and be ranked higher while the ranking of the less robust itemsets will
be more susceptible to noise.

In order to do this, we created from each dataset a synthetic dataset having
the same dimensions by sampling from a distribution. The underlying distribution 
had the same margins as the original data but otherwise items were independent.
We then mix the original data with the synthetic one, that is, an entry in a mixed
dataset is an entry from the synthetic dataset with the probability $\eta$, and
is an entry from the original dataset with the probability $1 - \eta$.
We tested two different noise levels $\eta = 0.05$ and $\eta = 0.1$.

We mined approximately $10\,000$ frequent closed itemsets from each original dataset. If the
dataset contained less than $10\,000$ itemsets, we set the threshold to one
transaction. Using the same thresholds we mined closed itemsets from the mixed
datasets. We sorted the itemsets using Algorithm~\ref{alg:closedestimate}. 

Let $X$ be an itemset ranked $i$th in the original data.  Assume that $X$ is ranked $j$th in the noisy data. We define compliance of $X$ by $1 / (\abs{i -
j} + 1)$. The compliance will be $1$ if $i = j$ and decreases to $0$ the
longer is the distance. The reason for using this particular definition is that
we can naturally set compliance to $0$ if $X$ is not found in the noisy data.
The compliances for top-$100$ itemsets are given in Figure~\ref{fig:clrank}. }

\begin{figure}[htb!]
\begin{center}
\subfigure[noise $\eta = 0.05$]{
\begin{tikzpicture}
\begin{axis}[xlabel={original itemset rank}, ylabel= {compliance},
    width = 6cm,
    xmax = 100, xmin = 1,
	ymax = 1, ymin = 0,
	no markers,
    cycle list name=yaf,
	xtick = {1, 21, 41, 61, 81, 100},
	ytick = {0, 0.1, ..., 1.1}
    ]
\addplot table[x expr = {\coordindex + 1}, y index = 0] {closed_ranks_0.05.dat};
\addplot table[x expr = {\coordindex + 1}, y index = 1] {closed_ranks_0.05.dat};
\addplot table[x expr = {\coordindex + 1}, y index = 2] {closed_ranks_0.05.dat};
\pgfplotsextra{\yafdrawaxis{1}{100}{0}{1}}
\legend{25\%, 50\%, 75\%}
\end{axis}
\end{tikzpicture}}
\subfigure[noise $\eta = 0.1$]{
\begin{tikzpicture}
\begin{axis}[xlabel={original itemset rank}, ylabel= {compliance},
    width = 6cm,
    xmax = 100, xmin = 1,
	ymax = 1, ymin = 0,
	no markers,
    cycle list name=yaf,
	xtick = {1, 21, 41, 61, 81, 100},
	ytick = {0, 0.1, ..., 1.1}
    ]
\addplot table[x expr = {\coordindex + 1}, y index = 0] {closed_ranks_0.1.dat};
\addplot table[x expr = {\coordindex + 1}, y index = 1] {closed_ranks_0.1.dat};
\addplot table[x expr = {\coordindex + 1}, y index = 2] {closed_ranks_0.1.dat};
\pgfplotsextra{\yafdrawaxis{1}{100}{0}{1}}
\legend{25\%, 50\%, 75\%}
\end{axis}
\end{tikzpicture}}

\end{center}
\caption{Rank compliance of an itemset in a noisy data as a function of robustness in the original data.
High compliance value imply that adding noise had little effect on the rank of
an itemset. Median and quartiles are computed over all datasets.}
\label{fig:clrank}
\end{figure}
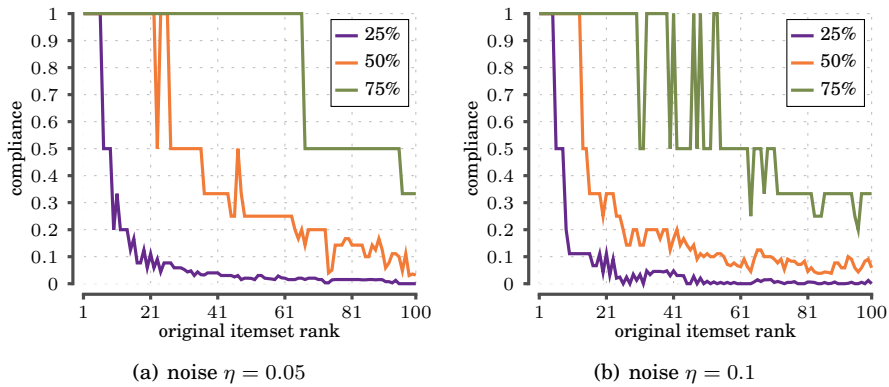

\rt{
From the figures we see that compliance stay high for robust itemsets and drop
as we move further down the original ranking. That is, the more robust
an itemset is, the less prone to noise it is. Adding more noise to the data implies
less compliance. For example, for noise level $\eta = 0.05$, top-60 itemsets
had a compliance of 0.25 or higher in half of the datasets. This means that their
rank changed only by $3$. On the other hand for noise level $\eta = 0.1$, top-50
itemsets had a compliance of 0.1 or higher in half of the dataset, in other
words, ranks changed by $9$.
}

\subsection{Ranking without $\alpha$}
Our next experiment was to compare the parameter-free ranking described in
Section~\ref{sec:order} against the rankings based on quantitative robustness
given specific values of $\alpha$.  We expect that rankings are similar for
large $\alpha$ values and difference increase when we lower $\alpha$. For
comparison we used the number of discordant pairs to calculate a distance of the rankings similar to Kendall's $\tau$. 
A discordant pair is a pair of itemsets $(X, Y)$ such that the first
method ranks $X$ higher than $Y$ and the second method ranks $Y$ higher than
$X$. We normalize the number of observed discordant pairs by $100b$, where $b$ is the  maximum number
of discordant pairs.
Hence, we  obtain a value between $0$ and $100$. If there are no ties in robustness, then 
$ b= N(N - 1) / 2$, where $N$ is the number of itemsets.
However if ties are presented, that is, the robustness induces a bucket order, then $b = N(N - 1) /2 -
\sum_{i = 1} B_i (B_i - 1) / 2 $, where $B_i$ is the size of each bucket, set of itemsets having the same robustness.
Values close to $0$ mean that rankings are in agreement. 

Typical examples are given in Table~\ref{tab:ranking} for the \emph{Mushroom} and \emph{Zoo} datasets, along with the 
averages taken over all datasets.
Surprisingly, the ranking distance is extremely small even for small values of $\alpha$ showing that the parameter free approach produces rankings similar to rankings under most $\alpha$. Starting at $\alpha = 0.5$ for \emph{Mushroom} and all $\alpha$ for \emph{Zoo} only about 1\% of pairs are discordant.
We see that values increase as we lower $\alpha$ which is expected since the parameter-free approach is based on large $\alpha$ values.

\begin{table}[htb!]
\tbl{Distance between parameter-free rankings and rankings based on $\alpha$ for \emph{Mushroom} and \emph{Zoo} datasets. Low values imply that rankings agree.
Value range is $0$--$100$.}{
\begin{tabular}{lrrrlrrrlrrr}
\toprule
& \multicolumn{3}{l}{\emph{Mushroom} $(\tau = 0.05)$} && \multicolumn{3}{l}{\emph{Zoo} $(\tau = 0.01)$} && \multicolumn{3}{l}{All datasets}\\
\cmidrule{2-4}
\cmidrule{6-8}
\cmidrule{10-12}
$\alpha$ & free & ts & nd &&  free & ts & nd &&  free & ts & nd \\
\midrule

$0.1$ & $0.30$ & $0.82$ & $0.39$ &  & $14.91$ & $7.26$ & $7.62$ &  & $4.35$ & $4.25$ & $2.80$\\
$0.2$ & $0.078$ & $0.27$ & $0.089$ &  & $10.01$ & $8.69$ & $4.31$ &  & $2.59$ & $2.67$ & $2.11$\\
$0.3$ & $0.017$ & $0.11$ & $0.044$ &  & $5.94$ & $5.40$ & $2.23$ &  & $1.46$ & $1.63$ & $1.45$\\
$0.4$ & $0.0016$ & $0.050$ & $0.015$ &  & $2.84$ & $2.77$ & $1.84$ &  & $0.69$ & $1.03$ & $0.91$\\
$0.5$ & $0.000032$ & $0.027$ & $0.0022$ &  & $1.12$ & $1.31$ & $1.11$ &  & $0.26$ & $0.56$ & $0.52$\\
$0.6$ & $0.0000050$ & $0.016$ & $0.0023$ &  & $0.32$ & $0.58$ & $0.56$ &  & $0.082$ & $0.25$ & $0.27$\\
$0.7$ & $0$ & $0.0020$ & $0.0011$ &  & $0.017$ & $0.20$ & $0.25$ &  & $0.013$ & $0.073$ & $0.11$\\
$0.8$ & $0$ & $0.0022$ & $0$ &  & $0$ & $0$ & $0.013$ &  & $0.000074$ & $0.013$ & $0.031$\\
$0.9$ & $0$ & $0$ & $0$ &  & $0$ & $0$ & $0$ &  & $0$ & $0.00049$ & $0.0039$\\

\bottomrule
\end{tabular}}
\label{tab:ranking}
\end{table}

\subsection{Top-k closed and free itemsets}

Closed itemsets are often used for tasks requiring interpretation of the itemsets,
because as maximum elements of an equivalence class they offer the most detailed description.
We
studied the highest ranked closed
itemsets for text datasets that are easily understood without domain
knowledge. As an illustrative example, we used the \emph{re0} news dataset from which we
mine $2493$ closed itemsets with minimum support $\tau = 0.05$. We ordered these itemsets
using the estimation technique given in Section~\ref{sec:orderclosed} and list the
top 45 itemsets in Table~\ref{tab:closed}.  The ranking is different from the one using
support, less frequent (but more robust) itemsets are commonly ranked higher that
frequent itemsets. For example, 'bank pct rate' occurs before the much more frequent itemset
'bank pct' showing that 'bank pct' is only closed in the full dataset due to relatively few 
documents using it without also using 'rate'.

\begin{table}[htb!]
\tbl{Top-45 closed itemsets from \emph{re0} ($\tau = 0.05$) dataset.}{
\begin{tabular*}{\columnwidth}{@{}r@{\extracolsep{\fill}}rr@{ }rrr@{ }rrr@{}}
\toprule
1. &  pct & 792 &
16. &  week & 310 &
31. &  canada & 117 \\
2. &  bank & 702 &
17. &  pct earlier & 127 &
32. &  pct month & 261 \\
3. &  trade & 485 &
18. &  japan & 318 &
33. &  econom & 295 \\
4. &  billion & 552 &
19. &  trade current & 126 &
34. &  billion dlr mln & 116 \\
5. &  market & 554 &
20. &  dlr & 472 &
35. &  told bank & 116 \\
6. &  billion dlr & 346 &
21. &  bank pct rate & 287 &
36. &  told nation & 116 \\
7. &  offici & 342 &
22. &  dollar & 336 &
37. &  pct japan & 115 \\
8. &  mln & 420 &
23. &  statem & 122 &
38. &  pct adjust & 115 \\
9. &  nation & 323 &
24. &  committe & 121 &
39. &  billion current & 115 \\
10. &  rate & 566 &
25. &  nation month & 121 &
40. &  european & 114 \\
11. &  bank market & 369 &
26. &  ministri & 120 &
41. &  month japan & 114 \\
12. &  foreign & 331 &
27. &  pct rise & 269 &
42. &  bank ad market & 114 \\
13. &  pct figur & 132 &
28. &  bank pct & 407 &
43. &  action & 114 \\
14. &  pct rate & 418 &
29. &  pct rate feb & 119 &
44. &  trade world & 114 \\
15. &  month & 391 &
30. &  lead & 118 &
45. &  nation japan & 114 \\
\bottomrule
\end{tabular*}}
\label{tab:closed}
\end{table}

Finally, we considered an alternative order by ranking itemsets based on how
free they are.  Note that a closed itemset is robust if the same transactions
cannot be explained by a superset whereas a free itemset is robust if the same
transactions cannot be explained by a subset. For example, a singleton $X$ will
be ranked higher than singleton $Y$ if $X$ has \emph{lower} support. The reason
for this is that it requires less transactions to be removed in order to make
$Y$ non-robust, namely the transactions not containing $Y$.  We present the
top-45 free non-singleton itemsets from \emph{re0} news dataset in
Table~\ref{tab:free}. \rt{These are frequent item pairs $ab$ such that $\supp{ab} \ll \supp{a}$
and $\supp{ab} \ll \supp{b}$, that is, a non-robust free item pair $ab$ would be such
that if we would remove a singleton $a$ (or $b$), then roughly the same transactions
will still cover the pattern.}
An example of
such non-robust free itemset is \emph{bank assist}. This itemset is ranked as 2\,465 out of
2\,558 itemsets. The support of this itemset is 96 but the support of \emph{assist} is
98, consequently there are only two documents in which \emph{assist} occurs but not
\emph{bank}.

\begin{table}[htb!]
\tbl{Top-45 free non-singleton itemsets from \emph{re0} ($\tau = 0.05$) dataset.}{
\begin{tabular*}{\columnwidth}{@{}r@{\extracolsep{\fill}}rr@{ }rrr@{ }rrr@{}}
\toprule
$1.$ & billion rate & $165$ & $16.$ & bank billion & $287$ & $31.$ & govern dollar & $92$\\
$2.$ & rate dlr & $132$ & $17.$ & billion pct & $288$ & $32.$ & foreign februari & $87$\\
$3.$ & trade rate & $146$ & $18.$ & rise dlr & $118$ & $33.$ & januari dlr & $81$\\
$4.$ & trade bank & $154$ & $19.$ & pct dlr & $210$ & $34.$ & monei dlr & $81$\\
$5.$ & billion market & $228$ & $20.$ & trade billion & $223$ & $35.$ & dollar februari & $89$\\
$6.$ & rate mln & $109$ & $21.$ & bank dlr & $211$ & $36.$ & rise japan & $78$\\
$7.$ & trade pct & $176$ & $22.$ & govern februari & $77$ & $37.$ & februari japan & $78$\\
$8.$ & bank pct & $407$ & $23.$ & govern mln & $90$ & $38.$ & dollar offici & $96$\\
$9.$ & market rate & $262$ & $24.$ & month dlr & $141$ & $39.$ & rise offici & $104$\\
$10.$ & trade mln & $130$ & $25.$ & trade dlr & $222$ & $40.$ & pct mln & $184$\\
$11.$ & market dlr & $186$ & $26.$ & pct market & $306$ & $41.$ & februari dlr & $94$\\
$12.$ & month mln & $106$ & $27.$ & market rise & $133$ & $42.$ & foreign mln & $95$\\
$13.$ & trade market & $203$ & $28.$ & februari offici & $83$ & $43.$ & nation februari & $89$\\
$14.$ & rise mln & $109$ & $29.$ & govern monei & $77$ & $44.$ & januari mln & $90$\\
$15.$ & trade rise & $115$ & $30.$ & januari govern & $77$ & $45.$ & monei month & $90$\\
\bottomrule
\end{tabular*}}
\label{tab:free}
\end{table}

\section{Discussion}\label{sec:discussion}
The experiments have shown that the number of itemsets can be largely reduced
on many datasets when requiring a certain robustness. The fact that the results
vary by dataset are another indication of the well known fact that itemset data
with different structures (dense vs. sparse, many items vs. many transactions)
behave very differently in mining tasks. 

We believe that robust itemsets can be beneficial for post-processing techniques such as \cite{bringmann09one} or \cite{vreeken11krimp} that use itemsets as their input and remove redundancy in the pattern set. Robust itemsets can be used as an alternative input reducing their runtime without sacrificing performance. Also, robust itemsets could be used instead of closed-itemsets as seeds to the AC-Close algorithm for approximate itemset mining \cite{cheng06acclose} improving its efficiency that was criticized in \cite{gupta08quantitative}.

The ranking of itemsets by robustness presents a new interestingness measure that can be used to choose the top-$k$ itemsets for interpretation or other data mining tasks. The intuition of robustness should be easy to understand for analysts but which ranking is better for specific data mining tasks remains to be studied. 

In particular it will be interesting to evaluate performance as features for classification tasks in contrast to direct mining of prediction tasks. For interpretable classifiers one would want itemsets to be long, thus use closed patterns. On the other hand the desire is for an itemset to be present in unseen data with high likelihood, so free itemset as the minimal elements of an equivalence class may generalize better. For both patterns we can ensure that they are present in many subsets of the training without actually sampling, potentially alleviating the need for nested cross validation.

\section{Summary}\label{sec:summary}
We have shown how robustness under subsampling for common classes of itemsets can be 
computed efficiently without actually sampling the data. The experimental results
show that the number of reported itemsets can be largely reduced on many datasets, 
in other words spurious itemsets that would not have been found in many subsets
of the data are removed. The approach can further be used to rank itemsets for 
top-$k$ mining by robustness. Future work will investigate the effect of using robust itemsets 
on data mining tasks such as clustering, classification, and rule generation using 
itemsets.

\bibliographystyle{acmsmall}
\bibliography{subsample}

\end{document}